\documentclass [12pt]{article}

\usepackage{amsmath,amsthm,amscd,amssymb}

\usepackage[small, centerlast, it]{caption}
\usepackage{epsfig}
\usepackage[titles]{tocloft}
\usepackage[latin1]{inputenc}
\usepackage{verbatim} 
\usepackage[active]{srcltx} 
\usepackage{fancyhdr}
\usepackage{caption}

\def\b1{{1\!\!1}}

\def\cG{{\ca G}}

\def\sH{{\mathsf H}}

\def\bC{{\mathbb C}}           

\def\bN{{\mathbb N}}

\def\bR{{\mathbb R}}

\def\bS{{\mathbb S}}

\def\bT{{\mathbb T}}
\def\bZ{{\mathbb Z}}

\def\gA{{\mathfrak A}}       
\def\gB{{\mathfrak B}}

\def\gS{{\mathfrak S}}

\def\beq{\begin{eqnarray}}
\def\eeq{\end{eqnarray}}

\newcommand{\ca}[1]{{\cal #1}}         

\usepackage{sectsty}
\sectionfont{\normalsize}
\subsectionfont{\normalsize\normalfont\itshape}
\newtheoremstyle{thm}
{12pt}
{12pt}
{\itshape}
{}
{\itshape\bfseries}
{}
{1em}
{}

\theoremstyle{thm}
\newtheorem{theorem}{Theorem}
\newtheorem{lemma}[theorem]{Lemma}
\newtheorem{proposition}[theorem]{Proposition}

\begin{document}

\hfill{\sl  UTM 751  Revised version, July 2012} 
\par 
\bigskip 
\par 
\rm


\par
\bigskip
\large
\noindent
{\bf   Mass operator and dynamical implementation of mass superselection rule}
\bigskip
\par
\rm
\normalsize


\noindent {\bf Eleonora Annigoni$^{a}$} and {\bf Valter Moretti$^{b}$}\\
\par
\noindent Department of  Mathematics, Faculty of Science, University of Trento, via Sommarive 14, 38050 Povo (Trento), Italy.
\smallskip

\noindent $^a$ eleonora.annigoni@gmail.com\qquad $^b$ moretti@science.unitn.it\\
 \normalsize

\par

\rm\normalsize


\rm\normalsize


\par
\bigskip

\noindent
\small
{\bf Abstract}.
We start reviewing Giulini's dynamical approach to Bargmann  superselection rule  proposing some improvements.  First of all we discuss some general features of the central extensions of the Galileian group used in Giulini's programme, in particular focussing on the interplay 
of classical and quantum picture, without making any particular choice for the multipliers.
Preserving  other features of Giulini's approach, we modify the mass operator of a Galilei invariant quantum system  to obtain a mass spectrum that is  
(i) positive and  (ii)  discrete, so  giving rise to  a standard (non-continuous) superselection rule. 
The  model results to be  invariant under time reversal  but  a further degree of freedom appears that  can be interpreted as describing an internal 
conserved charge of the system. (However,  adopting a POVM approach, the unobservable degrees of freedom can be pictured as a generalized observable 
automatically gaining a positive mass  operator without assuming the  existence of such a charge.) The effectiveness of Bargmann rule is shown to be
 equivalent to an averaging procedure over the unobservable degrees of freedom of the central extension of Galileian group.  Moreover, viewing the  Galileian 
 invariant quantum mechanics as a non-relativistic limit,  we prove that the above-mentioned averaging procedure giving rise to Bargmann superselection rule 
  is nothing but an effective de-coherence phenomenon due to time evolution if assuming that real measurements includes a temporal averaging procedure. 
  It happens when the added term $Mc^2$  is taken in due  account in the Hamiltonian operator since, in the  dynamical approach, the mass $M$ is an
  operator and cannot be trivially neglected
as in classical mechanics.  The presented results are quite general and rely upon the only hypothesis that the mass operator has point spectrum.
 These results explicitly show the 
interplay of the period of time of the averaging procedure, the energy content of the considered states, and the minimal difference of the mass operator 
eigenvalues.
\normalsize

\section{Introduction}
As is well know, in its elementary form, the celebrated superselection rule for the mass due to Bargmann states that,
for a quantum non-relativistic system a superposition of vectors representing pure states with different values of the total mass of the system 
 cannot represent a pure state of the system. While an incoherent superposition of permitted vectors is  an allowed mixture.\\
In \cite{Giulini},  Giulini has stressed that, comparing with other superselction rules,  Bargmann rule does
not have a well defined status. 
 This is because, in non-relativistic quantum physics, the mass enters the model of a physical system as a definitory quantity and not as an 
 observable that, depending on the state of the system, can or cannot be defined.  The standard motivation of Bargmann's superselection
  rule is based on the requirement of Galileian invariance. In Giulini's view, this approach does not permit
  to investigate the existence of any  dynamical phenomenon giving rise to the superselection rule.
  Starting from the classical Hamiltonian formulation, in \cite{Giulini}, Giulini has been able to re-interpret Bargmann's rule
   into a dynamical framework giving an explicit example of mass operator and assuming a central extension of the Galileian group
    as the physical group of symmetries of the system rather the Galileian one.
In that framework the superselection rule turns out to be equivalent to stating that, for some physical reason, 
   the extended system cannot be localized  with respect to the variable conjugated with the total mass operator.
  In  \cite{Giulini2}, it has been claimed that this non-localizability property arise from 
   some de-coherence process \cite{BGJKS}.  It is clear that the model has been  proposed as, and has to be 
  considered as, nothing but an elementary, not necessarily physical, example to illustrate a general viewpoint. However it is interesting 
  on its own right from a mathematical-physics viewpoint.
As a matter of fact, in this work we seriously consider Giulini's model for the mass operator in quantum non-relativistic theories. The reason
is twofold. On the one hand we would like to improve (from a pure mathematical-physics viewpoint)
 that model, getting rid of some difficulties (especially the  non-positive spectrum of the mass operator). On the other hand, accepting Giulini's 
  idea  that some superselection rules, formally obtained imposing invariance requirements,
  may actually arise from some real physical process,  we  intend to explictly discuss
  an elementary effective de-coherence process responsible for the mass superselection rule, relying upon the fact that 
  physical instruments take temporal average values.  
Regarding the first goal, changing the classical extended phase space with respect to that defined in \cite{Giulini}, we  indeed construct a 
classical model of the mass dynamical variable
whose quantization gives rise to a positive point-spectrum for the corresponding self-adjoint operator. In this way a
 further degree of freedom arises that may be interpreted as a conserved charge of the studied system of particles.
  Concerning the second goal we prove that, assuming 
  the existence of a mass operator,  the effective de-coherence process arises just to the 
  presence of the mass operator as an added term to the non relativistic Hamiltonian. That term is usually and incorrectly neglected  when performing the most simple
  non-relativistic approximation.
The mass operator leading to the mentioned result may have a different structure
with respect to that discussed in \cite{Giulini} and in the first part of this work, provided its spectrum 
 is a pure point-spectrum. 
In practice we prove that, if the temporal averaging process lasts an interval of time $T$ whose 
 order of magnitude is included in a certain interval $T_0 << T<< T_1$  determined by the total mass of the system and the energy content of the state, a vector $\Psi$ behaves  as a mixture of pure states allowed by the superselection rule, even if $\Psi$ is not permitted by the mass superselction rule when viewed as a pure state.
Referring to temporally averaged quantities over  longer periods of time,  the superselection rule holds, regardless the validity of the upper constraint $T<<T_1$,
as long as the scale of the mass (times $c^2$) is greater than  the energy content of the states, as it us usual in the macroscopic non-relativistic world.
   This effective de-coherence process, when referred to the model of mass operator introduced in the first part of the paper, can alternatively be
  interpreted as an averaging procedure on the unobservable part of the extended  symmetry group. \\
  The paper is organized as follows. After fixing the details of the studied physical system in Sec.2, we review Giulini's model from a slightly more mathematical 
  point of view both from classical and quantum level in Sec.3. This is done stating some results with a little more mathematical generality. Sec.4 is devoted to 
  improve the model getting rid of the negative part of the spectrum of the mass operator in particular. In Sec.5, after having focussed on some different approaches to the 
  superselection rule, we pass to  discuss the superselection rule as a physical process illustrated in Sec. 5.3. Conclusions are stated in Sec.6.
\section{Bargmann superselection rule and central extensions}\label{sec1}
As discovered by Bargmann \cite{Bargmann}, invariance under Galileian group for a quantum mechanical system
 entails that the system satisfies a superselection rule that imposes 
that any coherent superposition of states with different masses is forbidden.  
 Let us review  how it happens,
 also to fix notations and conventions used throughout.
\subsection{Non-relativistic quantum systems}
In the rest of the paper we deal with a general non-relativistic isolated quantum  system composed of $N$ particles, with masses $m_1,\ldots, m_N$,
described in an inertial reference frame equipped with orthonormal Cartesian coordinates. We shall employ these coordinates to define the Lebesgue measure of the various $L^2$ spaces. 
 The Hilbert space and the Hamiltonian  of the whole system are respectively:
  \beq \sH_0 =  \bigotimes_{k=1}^N \bC^{2s_k+1} \otimes L^2(\bR^3,d^3x_k)\:, \quad 
  \widehat{H}_0 = \overline{\sum_{k=1}^N \frac{\widehat{\bf p}_k^2}{2m_k}  + \widehat{V}(\{{\bf x}_i - {\bf x}_j\}})  \:.  \label{spaceH}\eeq
Above $\bC^{2s_k+1}$ and $L^2(\bR^3,d^3x_k)$ are the spin space 
and  is the orbital space of the $k$-th particle respectively,
 ${\bf x}_k= (x_{1k}, x_{2k}, x_{3k})$ and 
${\bf p}_k= (p_{1k}, p_{2k}, p_{3k})$
denote the positions and momenta  of the $k$-th particle referred to the mentioned coordinates.
A self-adjoint operator representing an observable associated with a corresponding classical quantity, $Q$, is  denoted with a hat: $\widehat{Q}$. 
We assume that the translationally invariant potential $V$ is also  rotationally invariant. 
Furthermore we henceforth suppose that the operator under the symbol of closure is essentially self-adjoint on  $C_0^\infty(\bR^{3N})$ and this is assured by Kato's theorem \cite{RS}, for a suitable choice of  $V$.
\subsection{Galileian group}
 Any  $g=(c_g, {\bf c}_g, {\bf v}_g, R_g)\in \bR\times \bR^3 \times \bR^3 \times SO(3) =: \cG$ acts as a Galileian transformation:
\begin{eqnarray}
(t', {\bf x}') = g (t,{\bf x})\quad \mbox{where} \quad
\left\{\begin{array}{cl} t' =& t + c_g \:,\\ {\bf x}' =& {\bf c}_{g}
+ t{\bf v}_{g}+
 \displaystyle R_g {\bf x}.\end{array}\right.\label{galileo}\end{eqnarray}
As a consequence $\cG$ becomes a group, the {\em Galileian group},  when equipped with the notion of product corresponding to the composition of coordinate transformations as above:
\beq (c_2, {\bf c_2}, {\bf v_2}, R_2) \cdot (c_1, {\bf c_1}, {\bf v_1},
R_1) =
\left(c_1+c_2,\: R_2 {\bf c_1}+ c_1{\bf v_2} + {\bf c_2}, \: R_2 {\bf
v_1}+ {\bf v_2},\: R_2R_1 \right)\:,\label{prodG}\eeq
with unit element and inverse, respectively given by
\beq e := (0, {\bf 0}, {\bf 0}, I) \quad \mbox{and}\quad  (c, {\bf c}, {\bf v}, R)^{-1} = (-c, R^{-1}(c{\bf v}-{\bf c}), -
R^{-1}{\bf v}, R^{-1})\:.\nonumber 
\eeq
Viewing the elements  $(c, {\bf c}, {\bf v}, R)$ of $\cG$ as $5\times 5$ real  matrices:
\begin{eqnarray}
 \left[
\begin{array}{ccc}
  R & {\bf v} & {\bf c}\\
  0 & 1 & c\\
  0 & 0& 1
\end{array}
\right]
\end{eqnarray}
the product in $\cG$  becomes the standard matrix product. Therefore $\cG$ is a subgroup of the matrix Lie group $GL(5,\bR)$ and, since $\cG$ is topologically closed in $\bR^{25}$, $\cG$ turns out to be a Lie group as well due to Cartan's theorem. With our hypothesis to consider $SO(3)$ instead the whole $O(3)$, $\cG$ results to be connected. 
It is evident from (\ref{prodG}) that $\cG$  is the iterated semi-direct product $\bR^4 \rtimes (\bR^3 \rtimes SO(3))$.

\subsection{Projective unitary representations of $\cG$ in quantum mechanics}
In view of Wigner theorem, the elements $g = (c_g, {\bf c}_g, {\bf v}_g, R_g)$ of the Lie group $\cG$ act, as active transformations, on the states of the system at $t=0$ in terms of unitary or antiunitary 
transformations $U_g$. Connectedness of $\cG$ assures that each $U_g$ is unitary. 
Up to a phase, the  explicit form of each $U_g$ is obtained
supposing that the induced active action on the observables  up to now mentioned,
$\widehat{A} \to U^*_g \widehat{A} U_g$, is that known from classical physics.
To this end we need another representation of the Hilbert space. We separate the  internal orbital degrees of freedom from the centre of mass degrees of freedom. The former ones  can be defined in terms of the translationally-invariant and boost-invariant Jacobi coordinates: ${\bf y}_k := {\bf x}_k - (\sum_{j=k+1}^N m_j {\bf x}_j)(\sum_{j=k+1}^N m_j)^{-1}$ for
$k=1,2,\ldots N-1$,
with conjugated translationally-invariant  and boost-invariant momenta ${\bf q}= (q_1,\ldots,q_{N-1})$.
From now on ${\bf X}= (X_{1}, X_{2}, X_{3})$ and ${\bf P} = (P_{1}, P_{2}, P_{3})$ denote the  position and momentum 
 of the centre of mass of the system, so that ${\bf P}$ coincides to the total momentum of the system, too.  The total mass is $M:= \sum_{k=1}^N m_k$. 
With these definitions, the Hilbert space $\sH$ turns out to be isomorphic to \beq 
 L^2(\bR^3, d^3X) \otimes L^2(\bR^{3N-3}, d^{3N-3}y) \bigotimes_{k=1}^N \bC^{2s_k+1} \label{spaceH2}\eeq
 and the Hamiltonian in (\ref{spaceH}) can be re-written employing the unitary transformation identifying $\sH$ to the Hilbert space in (\ref{spaceH2}):
 \beq
\widehat{H}_0 = \overline{\frac{\widehat{\bf P}^2}{2M} + 
\sum_{k=1}^{N-1} \frac{\widehat{\bf q}_k^2}{2\mu_k}  + \widehat{V}( \{{\bf y}_j\})}\:,\quad 
\mu_k := \frac{m_k \sum_{j=k+1}^N m_j}{m_k +\sum_{j=k+1}^N m_j}
\label{Ham2}
\eeq 
With these conventions, referring to  (\ref{spaceH2}) and (\ref{Ham2}),
 if all spins $s_k$ are integers \footnote{If the spin $s_k$ of some particle is half integer all the  discussion 
 may be re-cast,  replacing $SO(3)$ with its universal covering $SU(2)$ passing, as a consequence, to 
 deal with the universal covering of $\cG$.} one finds, explicitly writing the dependence on $M$,
\beq
U^{(M)}_g = Z^{(M)}_g \otimes  T_{(c_g,R_g)}  \bigotimes_{k=1}^N D^{2s_k+1}(R_g) \label{Ug}\:,
\eeq
where $SO(3) \ni R \mapsto D^{2s_k+1}(R)$ is the standard strongly-continuous irreducible unitary representation of $SO(3)$ for spin $s_k$ and
 $\bR \times SO(3) \ni (c,R)\mapsto T_{(c,R)}$ is the strongly-continuous unitary representation of the subgroup of time translations and rotations 
 on the internal orbital degrees of freedom
(that is in turn the product of the two unitary representations of the two commuting mentioned subgroups). 
Notice that the relevant group of time translations here is just that generated by\footnote{Employing Kato's theorem one can fix $V$ so that  the operator under the symbol of closure and the analogue  in  (\ref{spaceH}) (i.e. (\ref{Ham2})) are essentially self-adjoint on the respective space of smooth compactly supported functions.} $\overline{\sum_{k=1}^{N-1} \frac{\widehat{\bf q}_k^2}{2\mu_k}  +
 \widehat{V}}$ in  Eq.(\ref{Ham2}), that concerns the internal orbital degrees of freedom 
only.
Finally, the unitary transformation $Z^{(M)}_g : L^2(\bR^3, d^3X) \to L^2(\bR^3, d^3X)$, is better described  in momentum 
picture (at $t=0$)  as, for $\psi \in L^2(\bR^3, d^3X)$ and thus its Fourer-Plancherel transform $\widetilde{\psi} \in L^2(\bR^3, d^3P)$:
$$\widetilde{\psi}'({\bf P}) := \left(\widetilde{Z}^{(M)}_g \widetilde{\psi}\right) ({\bf P}) = 
e^{i(c_g{\bf v}_g -{\bf c}_g)\cdot ({\bf P}-M{\bf v}_g)}
 e^{i \frac{c_g}{2M}({\bf P}-M{\bf v}_g)^2 + iM\gamma_g}\widetilde{\psi}(R_g^{-1}({\bf P}- M{\bf v}_g))\:,$$
where $\gamma_g \in \bR$ can be fixed arbitrarily. The factor $M$ has been introduced for pure convenience and could be embodied in $\gamma_g$ itself.
In position picture we get
\beq
\psi'({\bf X}) := \left(Z^{(M)}_g \psi\right)({\bf X}) =  e^{iM({\bf v}_g \cdot {\bf X} + \gamma_g)} e^{ic_g \frac{{\bf P}^2}{2M}} 
 \psi\left( R_{g}^{-1}( {\bf X} -c_g {\bf v}_g - {\bf c}_g)\right)
\eeq
Applying the time evolution operator $e^{-it\frac{{\bf P}^2}{2M}}$ on both sides, the known formula arises:
\beq \psi'(t,{\bf X}) =   e^{iM({\bf v}_g \cdot {\bf X}- {\bf v}^2_g t/2 + \gamma_g)} \psi(g^{-1}(t,{\bf X})) \:.\label{popular}\eeq
If  $(t,{\bf x}) \in \bR \times \bR^3$, it is convenient to define for later purposes:
\beq
\chi_g(t, {\bf x}) := e^{i({\bf v}_g \cdot {\bf x} - {\bf v}^2_g t/2) + \gamma_g} \quad \mbox{and}\quad
\varphi_g(t, {\bf x}) := {\bf v}_g \cdot {\bf x} - {\bf v}^2_g t/2 + \gamma_g \label{phases}\:.
\eeq
Differently from $SO(3) \ni R \mapsto D^{2s_k+1}(R)$  and
 $\bR \times SO(3) \ni (c,R)\mapsto T_{(c,R)}$, the map $\cG \ni g \mapsto Z^{(M)}_g$ is not a linear representation, since 
 the group composition rule is not preserved. Rather,  it reads $Z^{(M)}_{g'}Z^{(M)}_{g} = \omega_M(g',g) Z^{(M)}_{g'g}$, for some phases $\omega_M(g',g)$.
 Since, instead,  $T$ and $D$ verify the composition rule, one finally finds, where it is apparent 
 that the mass $M>0$ enters the form of the phases just as an exponent:
\beq U^{(M)}_{g'}U^{(M)}_{g} = \omega(g',g)^M U^{(M)}_{g'g} \quad \mbox{with}\quad \omega(g',g) =: e^{i\xi(g',g)}\:.\label{proj}\eeq
From (\ref{galileo}), (\ref{popular}) and (\ref{phases}), this identity entails (where all functions are referred to the same choice for the constants $\gamma_g$):
\beq
\varphi_g(g'^{-1}(t,{\bf x}))  + \varphi_{g'}(t,{\bf x}) =  \varphi_{g'g}(t,{\bf x}) + \xi(g',g)\quad \mbox{if $g,g'\in \cG$ and $(t,{\bf x}) \in \bR\times \bR^3$.} \label{central}
\eeq
 Notice that (\ref{proj}) is valid regardless the choice of the explicit representation of the Hilbert space  (\ref{spaceH}) or  (\ref{spaceH2}). 
The existence of the phases $\omega(g',g)^M$, called {\em multipliers},  is not so surprising  because states are represented by rays, i.e. 
vectors up to a phase. We are, in fact,  dealing with  a (strongly-continuous)  unitary {\em projective} representation 
of $\cG$. As for every such representation, due to the associativity of the product of operators (thus the particular forms of 
$\chi_g$ and $\omega(g',g)$ do not matter), the multipliers $\omega^M$ verify the following identity:
 \beq
 \omega(g'' g')\omega(g''g',g) = \omega(g'',g'g) \omega(g',g)\quad \mbox{if  $g'',g', g \in \cG$}\:, \label{cociclo}
 \eeq
 or equivalently (working mod $2\pi$)
  \beq
 \xi(g'', g')+ \xi(g''g',g) = \xi(g'',g'g) + \xi(g',g)\quad \mbox{if  $g'',g', g \in \cG$}\:. \label{cociclo2}
 \eeq
 Another couple of general  identities, useful in the following, arise by multiplying $U_{g'}U_g = \omega(g',g) U_{g'g}$ and
 $U_{g^{-1}}U_{g'^{-1}} = \omega(g',g) U_{(g'g)^{-1}}$:
 \beq
 \omega(g', g)\omega(g^{-1},g'^{-1}) = \frac{\omega(g,g^{-1})\omega(g',g'^{-1})}{\omega(g'g, (g'g)^{-1})}\quad \mbox{if  $g', g \in \cG$}\:, \nonumber 
 \eeq
or equivalently
  \beq
 \xi(g', g)+ \xi(g^{-1},g'^{-1}) = \xi(g,g^{-1}) + \xi(g',g'^{-1}) - \xi(g'g, (g'g)^{-1})\quad \mbox{if  $g', g \in \cG$}\:. \label{cociclo4}
 \eeq

 \subsection{Bargmann superselection rule}
Since each  $U^{(M)}_g$ is  assigned up to a phase $\rho_g\in U(1)$ corresponding to any particular choice of the constant $\gamma_g$
in the exponent of $\chi_g$ in (\ref{phases}), it could be re-defined as $U'^{(M)}_g= \rho_g U^{(M)}_g$, obtaining corresponding new multipliers $\omega'_M(g',g)$. Obviously the action of $\cG$ on the space of the pure states (i.e. the rays of $\sH$) is not affected by the choice of 
$\rho_g$.
A natural question however  concerns the possibility to get $\omega'_M(g',g)=1$ for all $g,g'$ suitably  choosing the $\rho_g$.
A well-known sufficient co-homological condition assuring the feasibility of that choice (for continuous representations) is stated in a celebrated theorem by Bargmann \cite{BR}.  The theorem states that it is possible to fix the phases 
to eventually get trivial multiplicators $\omega'(g',g)=1$ for all $g,g'$, when the second co-homology group of the Lie algebra of the considered Lie group is trivial\footnote{Actually the theorem applies to 
the continuous unitary projective representations of 
connected simply-connected Lie groups, so  the universal covering of $\cG$ has to considered instead of $\cG$ itself. It is obtained making use of the covering homomorphism $SU(2) \mapsto SO(3)$.}. As a matter of fact, $\cG$ does not verify that hypothesis, differently from the Poincar\'e group. However Bargmann's requirement is just sufficient and not necessary, therefore a closer investigation is necessary.\\
If no choice of the phases assures $\omega'_M(g',g)=1$ for all $g,g'$, a problem may appear whenever one wishes to consider a system {\em where the total mass} $M$ can assume {\em different} values, because  $\omega^M$ in (\ref{proj}) depends on the mass $M$ of the system.
Consider a vector  $\psi = \psi_M + \psi_{M'} \in \sH$ that is a superposition of two vectors
$\psi_M$,  $\psi_{M'}$ with different masses $M\neq M'$. The action of $\cG$ implies the appearance of {\em different} phases $\omega(g',g)^M\neq \omega(g',g)^{M'}$ in front of $\psi_M$ and $\psi_{M'}$ so that a {\em relative} phase appears.
That does not permit to define a unitary {\em projective} representation of $\cG$ on the full space, because the composition rule of $\cG$ fails to be preserved even taking the quotient respect to overall phases.  
Once again, to get rid of the relative phases,  one may hope to get  $\rho'_g U^{(M')}_g= \rho_g U^{(M)}_g$ for all $g\in \cG$, with a suitable choice of the phases  $\rho_g$  and $\rho'_g$ (i.e. of the constants $\gamma_g$).
Nevertheless \cite{Bargmann} Bargmann proved (see the appendix)  that, if $M,M'>0$ and $M\neq M'$,
whatever choice one makes for the phases $\rho_g$ and $\rho'_g$,  both:

(1) $\omega'_M(g',g)=1$ for all $g,g'$ cannot hold,

 (2) $\rho'_g U^{(M')}_g= \rho_g U^{(M)}_g$ for all $g\in \cG$ cannot hold.
 
 \noindent This no-go result leads to the celebrated Bargmann superselection rule \cite{Bargmann}:\\
 
\noindent {\bf Mass Superselection Rule}: {\em For a quantum non-relativistic system, thus admitting the Galileian group
 as a group of symmetries, a superposition of vectors representing pure states with different values of the total mass of the system 
 cannot represents a pure state of the system. (While an incoherent superposition of permitted vectors can represent an allowed mixture.)}

\subsection{Central extensions} We can view the unitary projective representation 
$\cG \ni g \mapsto U^{(M)}_g$ as a part of a proper unitary representation, provided we enlarge $\cG$ to a {\em central extension}.\\
 One starts from the product $U(1) \times \cG$
or $\bR\times \cG$  and defines therein the product, respectively:
\beq
(\chi',g') (\chi,g) = (\chi'\chi \omega(g',g), g'g)\:, \quad (r',g') (r,g) = (r+r'+  \xi(g',g), g'g)\:.
\eeq
Thanks to, respectively, (\ref{cociclo}) and (\ref{cociclo2}) (assumed to be exactly and not only  mod $2\pi$)
one easily verifies that these products endow, respectively,  $U(1) \times \cG$ and $\bR \times \cG$ with a group structure where, respectively, $U(1)$ and $\bR$ are included in the centre. 
These, respectively, $U(1)$ and $\bR$ {\em central extensions} of $\cG$,  have a natural structure of Lie groups if the functions $\omega$ and $\xi$ are sufficiently regular \cite{BR}. We indicate these groups with, respectively, $U(1) \times_\omega \cG$ and
$\bR \times_\xi \cG$.
It is evident that $U(1) \times_\omega \cG \ni (\chi, g) \mapsto \chi^M U^{(M)}_g$
and $\bR \times_\xi \cG \ni (r, g) \mapsto e^{iMr} U^{(M)}_g$ are unitary representations of the considered central extensions and they reduce to the unitary projective  representation $\cG \ni g \mapsto U^{(M)}_g$ when restricted, respectively, to the elements $(1,g)$
and $(0,g)$, $g\in \cG$. All that is well known in the literature and not only referring to Galileian group.  \\
Concerning the Galileian group, we can pre-announce that one of the ideas in \cite{Giulini, Giulini2} was that 
actually the true physical group of transformations is not $\cG$, but it is one of its central extensions.

\section{Giulini's model revisited} 
 Giulini \cite{Giulini} has stressed that Bargmann rule has a not well-defined status (comparing with other superselction rules). 
 This is due to the fact that, in non-relativistic quantum physics, the mass appears as a definitory quantity for a given  system and not as an 
 observable that, depending on the state of the system, can or cannot be defined. In Giulini's view, this approach does not permit
  to investigate the existence of any  dynamical phenomenon giving rise to the superselection rule.
   Giulini has been able to re-interpret Bargmann's rule
   into a dynamical framework, giving an explicit example of mass operator and assuming $\bR \times_\xi \cG$ as physical group of the system 
   instead of $\cG$.  The superselection rule is equivalent to stating that, for some physical reason to be investigated, 
   the extended system cannot be localized  with respect to the variable conjugated with the total mass operator.
  In  \cite{Giulini2}, it has been claimed that this non-localizability property could be  due to
   some de-coherence process \cite{BGJKS}.  It is clear that the model has to be 
  considered just as a elementary, not necessarily physical, example to illustrate a general viewpoint. However it is interesting 
  on its own right from a mathematical-physics viewpoint.
In the following we shall  review this model from a slightly more abstract viewpoint.

\subsection{Classical system}\label{secabcbde}
Let us start from  the classical Hamiltonian  action:
\beq 
S_0 = \int_{t_1}^{t_2}   \left[ \sum_{k=1}^{N} {\bf p}_k \dot{\bf x}_k  - H_0(\{{\bf p}_j, {\bf x}_j\}) \right] dt\label{action} \:.
\eeq
The idea is to define masses $m_k \in \bR$  as new canonical variables whose equation of motion is trivial, so that they are constants of motion. 
The simplest way to do it is thinking of the masses $m_k$ as momenta  of  corresponding conjugated variables  $\zeta_k \in \bR$, without modifying the Hamiltonian function.
As it does not contain any $\zeta_k$, the variables  $m_k$ turn out to be constant of  motion  as wished. The new action reads
(where $H$ is the same as $H_0$ but considering the $m_k$ as dynamical variables):
  \beq
  S= \int_{t_1}^{t_2}   \left[ \sum_{k=1}^{N} {\bf p}_k \dot{\bf x}_k   +\sum_{k=1}^{N} m_k \dot{\zeta}_k - 
  H(\{ m_j,  {\bf p}_j, {\bf x}_j\})\right] dt\:.
\eeq
Next we pass to define  the action of Galileian group on the extended phase space 
$\bR^{6N+2N}$,
including 
all variables $\{{\bf x}_k, {\bf p}_k, {\bf \zeta}_k,  m_k\}_{k=1,\ldots,N}$.  The idea is to require that:
 
 (a) $\cG$ acts in the standard way on $({\bf x}_k, {\bf p}_k)$; 
 
 (b) it leaves  the masses $m_k$ fixed;
 
 (c) $\cG$ acts on each $\zeta_k$ in the same way regardless the value of $k$; 
 
 (d) the subgroup of time translations acts in the standard way with respect to the (autonomous) Hamiltonian flow;
 
 (e) the action $S'$ remains invariant under the action of the remaining one-parameter subgroup of $\cG$.\\
\noindent  {\em We stress that, in this way,  the solutions of Hamilton equations are transformed to such solutions, so that  each $g \in \cG$ defines a symmetry of the given dynamical system}. \\ With these constraints one easily finds that, indeed every $g\in \cG$ acts on the phase space by means of a time-dependent {\em canonical transformation}
$F_{t,g}$. As a matter of fact, if $f^{(c)}(t):= f(t-c)$ denotes the standard time-translation along the Hamiltonian flow and
$g=(c_g, {\bf c}_g, {\bf v}_g, R_g)$, we have:
 \beq
F_{t,g} :  \begin{bmatrix} {\bf x}_k(t)\\  {\bf p}_k(t)\\ \zeta_k(t), \\ m_k(t)
\end{bmatrix} 
\mapsto
\begin{bmatrix} {\bf c}_g + (t-c_g){\bf v}_g + R_g{\bf x}^{(c)}_k(t) \\ 
 m_k{\bf v}_g+R_g{\bf p}^{(c)}_k(t) \\ 
 \zeta^{(c)}_k(t)- {\bf v}_g\cdot R_g {\bf x}^{(c)}_k(t) - {\bf v}_g^2(t-c_g)/2 + \Gamma_{g}\\
  m_k(t)\:,
\end{bmatrix} \label{Ftg}
\eeq
$\Gamma_{g}$ is an undetermined constant. 
Here a problem similar to that found in the quantum side arises: The map $\cG \ni g \mapsto F_{t,g}$ is not a representation of $\cG$, since the composition rule of the group does not work: in general,  
$F_{t,g'}\circ F_{t,g} \neq F_{t,g'g}$.  \\
The idea is to think of $F_{t,g}$ as a part of a true representation of a larger group than $\cG$, more precisely a $\bR$-central extension. 
Indeed the following result holds (it generalizes, from a mathematical viewpoint, a corresponding statement in 
Sec.2 of  \cite{Giulini}, as we do not fix any particular form for the multipliers).

\begin{proposition}\label{p1} Consider the central extension $\bR \times_\xi \cG$, where $\xi$ is that associated with 
the functions $\varphi_g$ as in (\ref{phases}) determined by an arbitrary  choice of the constants $\gamma_g$.\\
The class  of transformations
on the extended phase space  $F_{t, (r,g)} :\bR^{6N+2N} \to 
\bR^{6N+2N}$
 \beq
F_{t,(r,g)} :  \begin{bmatrix} {\bf x}_k(t)\\  {\bf p}_k(t)\\ \zeta_k(t), \\ m_k(t)
\end{bmatrix} 
\mapsto
\begin{bmatrix} {\bf c}_g + (t-c_g){\bf v}_g + R_g{\bf x}^{(c)}_k(t) \\ 
 m_k{\bf v}_g+R_g{\bf p}^{(c)}_k(t) \\ 
\zeta^{(c)}_k(t) + \varphi_{g^{-1}}(t-c_g, {\bf x}^{(c)}_k(t)) - r -\xi(g^{-1},g)\\
  m_k(t)
\end{bmatrix} \nonumber 
\:,
\eeq
is a representation of $\bR \times_\xi \cG$ in terms of time-dependent canonical transformations. 
In other words, for every $t\in \bR$:on

(i)  $F_{t,(r,g)}$ is canonical if $(r,g) \in \bR \times_\xi \cG$,

(ii) $F_{t,(r',g')}F_{t,(r,g)}
= F_{t,(r',g')(r,g)}$ if $(r',g'), (r,g) \in \bR \times_\xi \cG$.\\
Finally,  
$F_{t,(0,g)} = F_{t,g}$ in (\ref{Ftg}) when choosing $\Gamma_g := -\xi(g^{-1},g) + \gamma_g$ therein for all $t\in \bR$ and $g\in \cG$.
\end{proposition}
\noindent Before going on with the proof we notice that, varying 
 $(r,g)\in \bR \times \cG$, the class of functions $\{F_{t,(r,g)}\}_{t\in \bR}$ turns out to be 
one-to-one associated with the corresponding transformation acting in the extended phase-space  $F_{(r,g)} : \bR \times \bR^{6N+2N} 
 \to \bR \times \bR^{6N+2N}$:
  \beq
F_{(r,g)} :  \begin{bmatrix} t\\
{\bf x}_k\\  {\bf p}_k\\ \zeta_k, \\ m_k
\end{bmatrix} 
\mapsto
\begin{bmatrix} t'\\
{\bf x}'_k\\  {\bf p}'_k\\ \zeta'_k, \\ m'_k
\end{bmatrix}  := 
\begin{bmatrix} t + c_g\\ {\bf c}_g + t{\bf v}_g + R_g{\bf x}_k \\ 
 m_k{\bf v}_g+R_g{\bf p}_k \\ 
 \zeta_k +\varphi_{g^{-1}}(t, {\bf x}_k)  - r -\xi(g^{-1},g)\\
  m_k
\end{bmatrix}\:. \label{Frg}
\eeq

\noindent {\em Proof of Proposition \ref{p1}}. The fact that each transformation $F_{t,(r,g)}$ is canonical and the last statement 
 can be verified by direct inspection (the former by computing the time-constant Jacobian matrices of the considered transformations 
 and proving that they preserve the symplectic unit matrix). Let us prove (ii) referring to (\ref{Frg}). We notice that  $F_{t,(r',g')}F_{t,(r,g)}
= F_{t,(r',g')(r,g)}$ for all $t$ is equivalent to $F_{(r',g')}F_{(r,g)}
= F_{(r',g')(r,g)}$. 
 Let us prove that, in fact, the latter holds true. Consider, in (\ref{Frg}),
the transformation law for the components $\zeta_k$:
$$\zeta_k \mapsto \zeta'_k := \zeta_k + \varphi_{g^{-1}}(t,{\bf x}_k)- r - \xi(g^{-1},g)\:.$$
When composing two transformations $F_{(r',g')}$ and $F_{(r,g)}$, all components of the vectors in (\ref{Frg}) transform 
in agreement with the composition rule of $\cG$ (and the $m_k$ remain fixed), barring the components $\zeta_k$ that transform as:
$$\zeta_k \to \zeta_k'' = \zeta_k + \varphi_{g^{-1}}(t, {\bf x}_k) - \xi(g^{-1},g) + \varphi_{g'^{-1}}(g(t, {\bf x}_k)) - \xi(g'^{-1},g') - (r+r')\:.$$
Exploiting (\ref{central}) and (\ref{cociclo4}) with $g$ replaced for $g'^{-1}$ and $g'$ replaced for
$g^{-1}$, one can re-write it as
$$\zeta_k \to \zeta_k'' = \zeta_k + \varphi_{(g'g)^{-1}}(t, {\bf x}_k) - (r+r') 
- \xi((g'g)^{-1},g'g) \:.$$
It immediately implies
$F_{(r',g')}F_{(r,g)}
= F_{(r',g')(r,g)}$ as wished.\\
$\Box$

\subsection{Quantization}
The system described in the previous section can straightforwardly be quantized {\em \`a la Dirac}, taking into account the fact that 
$\zeta_k, m_k$ are pairs of canonical variables defined on the whole real axis and thus corresponding to canonically conjugated operators 
$\widehat{\zeta}_k, \widehat{m}_k$. The complete Hilbert space of the system is now extended to 
$$\sH:=  L^2(\bR^{N}\times \bR^{3N}, d^{N}m\otimes d^{3N} x) \simeq
\int_{\bR^N}^\oplus  \sH_{\{m_k\}} d^N m\:,\quad  \sH_{\{m_k\}} := L^2( \bR^{3N}, d^{3N} x)$$
where we see the Hilbert space of the system as a direct integral of the Hilbert spaces at fixed masses $m_k$ with respect to the Lebesgue measure on 
$\bR^n$. 	We are in a position to state a second proposition that makes more general, from a mathematical viewpoint, a corresponding 
statement in Sec.3 of  \cite{Giulini} (since we do not fix any particular form for the multipliers). The item (ii) shows the link 
between the unitary representation of the central extension $\bR \times_\xi \cG$ and the canonical representation of the central extension 
discussed in the previous proposition. Now, differently form the corresponding identity (\ref{popular}), the phases are re-absorbed because we are working with a (central) extension of $\cG$ instead of $\cG$ itself.

\begin{proposition}\label{p2} Consider the central extension $\bR \times_\xi \cG$, where $\xi$ is that associated with 
the functions $\varphi_g$ as in (\ref{phases}) determined by an arbitrary  choice of the constants $\gamma_g$ and
focus on the class $\{U_{(r,g)}\}_{(r,g)\in \bR \times_\xi \cG}$ of transformations on the extended Hilbert space (viewed as a direct integral)
$U_{(r,g)} : \sH \to \sH$, where $U_{t}(\{m_k\})$ is the time evolution operator in the fiber $\sH_{\{m_k\}}$ of the direct integral:
\beq
\left(U_{(r,g)} \Psi\right) (\{m_k\}, \{{\bf x}_k\})  = e^{i \sum_{j=1}^n({\bf v}_g \cdot {\bf x_j} + \gamma_g +r)m_j} U_{-b_g}(\{m_k\})
 \Psi\left(\{m_k\},  \{R_g^{-1}({\bf x}_k -c_g {\bf v}_g -{\bf c}_g)\}\right)\:.
\eeq
They satisfy the following.

(i)  $\bR \times_\xi \cG \ni (r, g) \mapsto U_{(r,g)}$ is a unitary representation of $\bR \times_\xi \cG$ that, for $r=0$, reduces to the unitary projective representation (\ref{Ug}) in each fibre $\sH_{\{m_k\}}$,
associated with the multiplier $e^{i M \xi(g',g)}$ with the total  mass $M= \sum_{j=1}^N m_j$.

(ii) Passing in $\zeta$-picture,  $\check{\Psi}' (\{\zeta_k\}, \{{\bf x}_k\})  :=  \left(\check{U}_{(r,g)} \check{\Psi}\right) (\{\zeta_k\}, \{{\bf x}_k\})$
and taking into account the time evolution, one has for $\Psi \in \sH$
\beq
\check{\Psi}' (t, \{\zeta_k\},\{{\bf x}_k\})  = \check{\Psi}(\check{F}_{(r,g)^{-1}}(t, \{\zeta_k\}, \{{\bf x}_k\})) 
\eeq
where $\check{F}_{(r,g)}$ is the function that considers only the first, the second and the fourth
components  of the function (\ref{Frg}):
 \beq
\check{F}_{(r,g)} :  \begin{bmatrix} t\\
{\bf x}_k\\  {\bf p}_k\\ \zeta_k, \\ m_k
\end{bmatrix} 
\mapsto
\begin{bmatrix} t + c_g\\ 
  {\bf c}_g + t{\bf v}_g + R_g{\bf x}_k\\ 
 \zeta_k +\varphi_{g^{-1}}(t, {\bf x}_k)  - r -\xi(g^{-1},g)
\end{bmatrix}\:. 
\eeq
\end{proposition}

\begin{proof} The first statement easily arises from the definition of direct integral, taking into account (a) that  $U_{(r,g)}$ is of decomposable type 
and it  reduces to unitary operator $e^{i M r}U_g^{(M)}$
on every fibre $\sH_{\{m_k\}}$ with the total  mass $M= \sum_{j=1}^N m_j$, and (b) that (\ref{proj}) 
holds therein. Concerning (ii), following the same route to reach the identity (\ref{popular}), we have:
$$\Psi' (\{m_k\}, t, \{{\bf x}_k\})
= e^{i\sum_k m_k(\varphi (t, {\bf x}_k) + r)}\Psi (\{m_k\}, g^{-1}(t,\{{\bf x}_k\})) \:,$$
where $t$
(whose position is temporarily chosen as follows for sake of convenience)
indicates the time evolution.
Therefore:
$$\check{\Psi}'(\{\zeta_k\}, t, {\bf x}_k)
= \frac{1}{(2\pi)^{N/2}}\int_{\bR^N} e^{i\sum_k m_k(\zeta_k +\varphi_g (t, {\bf x}_k) + r)}\Psi (\{m_k\}, g^{-1}(t,\{{\bf x}_k\})) d^Nm$$
$$= \check{\Psi}(\{\zeta_k +\varphi_g (t, {\bf x}_k) + r\},  g^{-1}(t,\{{\bf x}_k\}))\:.$$
The whole argument of $\check{\Psi}$ in the right-hand side, restoring the standard position of the variable $t$, can be written
$\check{F}_{(r,g)^{-1}}(t, \{\zeta_k\}, \{{\bf x}_k\})$, paying attention to the identities
$(r,g)^{-1} = (-r - \xi(g^{-1},g), g^{-1})$ and $\xi(g^{-1},g)= \xi(g, g^{-1})$.
\end{proof}
\noindent Within this minimal but nice approach, a total mass operator $\widehat{M}$ exists, it is nothing but the multiplicative operator $\sum_{k=1}^N m_k$ in the mass representation.
Therefore the superselection  rule of the mass can be coherently formulated and discussed. Similarly to the couple of conjugated
 variables $\widehat{\bf P}_k, \widehat{\bf  X}_k$,  one can define a corresponding couple $\widehat{M}, \widehat{Z}$ 
  where the latter is the quantum observable associated with the center of mass for the $\zeta_k$ coordinates. As stressed in \cite{Giulini,Giulini2}, 
  the superselection rule can be stated by requiring 
  that localizability in the variable $\widehat{Z}$ is forbidden since the allowed states are sharply localized with respect to the variable $\widehat{M}$. 
  The general suggestion by 
  Giulini is the following.``It seems plausible that many derivations of superselection rules from
purely formal arguments in fact make at least one contingent physical assumption
of that sort. For better understanding the actual physical input one should in our
opinion 1.) find the right dynamical theory in which the relevant quantities are
manifestly dynamical and 2.) address the question of what is actually measurable
within that framework.''

\section{Improving Giulini's model}
Though, as clearly stressed in \cite{Giulini}, the presented model should not be considered as a realistic physical model,  but just as a minimal example to 
illustrate  a general, more physical, 
approach to superselection rules. Here we shall try to take it seriously. This will be done in order to improve it  and  go ahead with Giulini's general programme. \\
Within the illustrated framework, Bargmann's rule should  be discussed with the appropriate theoretical machinery and one expects to find the following features.
As we shall see in Sec.\ref{secquasifine}, any  superselction rule implies, on the one hand that  the Hilbert space is
decomposed into an orthogonal direct sum of coherent sectors: 
$\sH = \sum_{m\in \sigma(\widehat{M})} \sH_m$ where $\sH_m$ is the eigenspace of $\widehat{M}$ (or of the observable inducing the superselection rule) with eigenvalue 
$m$. 
On the other hand, the von Neumann algebra of (bounded) observables\footnote{The observables in
 $\gA$ are obviously only the self-adjoint operators therein. Nevertheless we shall not insist on this distinction in the rest of the paper unless strictly necessary.} 
 allowed by Bargmann rule is not a factor because it turns out to be $\gA := \{P_m \:|\: m \in \sigma(\widehat{M})\}'$,  
 $P_m$ being  the orthogonal projector on $\sH_m$, that consequently belong to the center of $\gA$.
This framework does not seem to be completely appropriate for the constructed model, because  the found mass operator  has purely continuous spectrum so that 
no proper eigenspaces exist.
 This is not an  insurmountable difficulty since, as already stressed in \cite{Giulini,Giulini2}, superselection rules can be even formulated  in reference to an observable with (a part of) spectrum of 
 continuous type, decomposing 
 the  Hilbert space in terms of a direct integral over the spectrum itself.  Then the von Neumann algebra $\gA$ of (bounded) observables 
 allowed by the superselection rule is made of all the bounded operators commuting with the spectral measure of that observable.
 The  most annoying feature, arising in this way, is that the pure states of the system (the extremal states that must exist in view of Krein-Milman theorem)
describing states where  the observable is defined, cannot be, in general, represented by vectors of the Hilbert space.
 In general they are only algebraic states on $\gA$. In our case thus mass-defined states 
 are not representable by state vectors of $\sH$. One may adopt several viewpoints in this situations, more o less mathematically or physically minded.
 We shall not address this issue here, rather we shall try to turn the continuous spectrum into a point-wise one.\\
 A second, much  more serious problem, is that the spectrum of $\widehat{M}$ encompasses negative values. This is forbidden in physics and the problem has to be
 tackled and solved necessarily.  (Sticking to the case of continuous values of the mass, one could tray to solve this problem by assuming $m_k= e^{\mu_k}$ with $\zeta_k\in \bR$ conjugated to $\mu_k\in \bR$.) \\
 A third, quite minor, difficulty is a mismatch with the standard time-reversal operation. As $\zeta_k$ plays the r\^ole of a coordinate and $m_k$ 
 that of a conjugated momentum, one expects that, under time reversal, $\zeta_k(t) \to \zeta_k(-t)$ and $m_k(t) \to -m_k(-t)$.  Actually it should be in contradiction 
 with the equation of motion of $\zeta_k$, at least for $V$ not depending on $m_k$:
 $$\frac{d\zeta_k}{dt} = - \frac{{\bf p}_k^2}{2m^2_k}\:.$$
 We notice that the problem can be formally solved by interchanging the r\^ole of the two phase space coordinates, so that the time reversal would be $\zeta_k(t) \to -\zeta_k(-t)$ and $m_k(t) \to m_k(-t)$.\\
 In the rest of the paper, we shall improve the model along the indicated directions, solving all problems at the same time. Moreover we propose another perspective to interpret the superselection rule in term 
 of an averaging procedure on the unobservable part of the central extension of $\cG$. Eventually, we prove that the averaging process, and thus the mass superselection
 itself, can be interpreted in physical terms, just through a dynamical process as argued by Giulini, provided one takes relativistic corrections into account.

\subsection{Modified classical system}
We modify the  action (\ref{action}) again, but we perform two changes: (i) each coordinate $\zeta_k$
describes a circle, i.e. $\zeta_k \in (-\pi,\pi]$ with identified endpoints, (ii) it holds $m_k= |n_k|$
where $n_k\in \bR$ is the conjugated momentum of $\zeta_k$. In this case the masses are evidently positive and the mismatch with the time-reversal operation disappears.
On the other hand, a smoothness singularity arises for $n_k=0$. However a more severe singularity was already present in Giulini's model (and it appears in ours as well)
since the Hamiltonian (in both cases) is singular for $m_k=0$. The action is now:
  \beq
  S= \int_{t_1}^{t_2}   \left[ \sum_{k=1}^{N} {\bf p}_k \dot{\bf x}_k   +\sum_{k=1}^{N} n_k \dot{\zeta}_k - 
  H(\{|n_j|,  {\bf p}_j, {\bf x}_j\})\right] dt\:.
\eeq
As before  we pass to define  the action of Galileian group on the extended phase space 
$\bR^{6N} \times \bT^{N} \times \bR^N$,
including 
all variables $\{{\bf x}_k, {\bf p}_k, {\bf \zeta}_k, n_k\}_{k=1,\ldots,N}$.  Again the idea is to impose 
the same requirements (a),(b),(c),(d) and (e) as in Sec. \ref{secabcbde}; so that, in particular, the solutions of Hamilton equations are transformed to such solutions and consequently each $g \in \cG$ defines a symmetry of the given dynamical system. As before the masses are constants of motion as expected.  With these constraints one easily finds again that indeed every $g\in \cG$ acts on the phase space by means of a time-dependent  canonical transformation
$G_{t,g}$. If
$g=(c_g, {\bf c}_g, {\bf v}_g, R_g)$,
 \beq
G_{t,g} :  \begin{bmatrix} {\bf x}_k(t)\\  {\bf p}_k(t)\\ \zeta_k(t), \\ n_k(t)
\end{bmatrix} 
\mapsto
\begin{bmatrix} {\bf c}_g + (t-c_g){\bf v}_g + R_g{\bf x}^{(c)}_k(t) \\ 
 |n_k|{\bf v}_g+R_g{\bf p}^{(c)}_k(t) \\ 
 \zeta^{(c)}_k(t)+\mbox{sign}(n_k)(- {\bf v}_g\cdot R_g {\bf x}^{(c)}_k(t) - {\bf v}_g^2(t-c_g)/2 + \Gamma_{g}) \:\: \mbox{\em mod $2\pi$}\\
  n_k(t)\:,
\end{bmatrix} \label{Ftg2}
\eeq
$\Gamma_{g}$ is an undetermined constant. As before,  the map $\cG \ni g \mapsto G_{t,g}$ is not a representation of $\cG$ since the composition rule of the group fails to be preserved. 
Nevertheless $G_{t,g}$ can be viewed as a part of a true representation of a $U(1)$-central extension of $\cG$.

\begin{proposition}\label{p12} Consider the central extension $U(1) \times_\omega \cG$, where the $\omega(g',g) = e^{i\xi(g',g)}$ are associated with 
the functions $\varphi_g$ as in (\ref{phases}) determined by an arbitrary  choice of the constants $\gamma_g$.\\
The class  of transformations
on the extended phase space  $G_{t, (\chi,g)} :\bR^{6N}\times \bT^N \times \bR^N \to 
\bR^{6N}\times \bT^N \times \bR^N$ and where $\chi= e^{ir}$:
 \beq
G_{t,(\chi,g)} :  \begin{bmatrix} {\bf x}_k(t)\\  {\bf p}_k(t)\\ \zeta_k(t), \\ n_k(t)
\end{bmatrix} 
\mapsto
\begin{bmatrix} {\bf c}_g + (t-c_g){\bf v}_g + R_g{\bf x}^{(c)}_k(t) \\ 
 |n_k|{\bf v}_g+R_g{\bf p}^{(c)}_k(t) \\ 
\zeta^{(c)}_k(t) + \mbox{\em sign}(n_k)(\varphi_{g^{-1}}(t-c_g, {\bf x}^{(c)}_k(t)) - r -\xi(g^{-1},g))\:\: \mbox{\em mod $2\pi$}\\
  n_k(t)
\end{bmatrix} 
\:,
\eeq
is a representation of $U(1)\times_\omega \cG$ in terms of time-dependent canonical transformations. 
In other words, for every $t\in \bR$:

(i)  $G_{t,(\chi,g)}$ is canonical if $(\chi,g) \in U(1) \times_\xi \cG$,

(ii) $G_{t,(\chi',g')}G_{t,(\chi,g)}
= G_{t,(\chi',g')(\chi,g)}$ if $(\chi',g'), (\chi,g) \in U(1) \times_\xi \cG$.\\
Finally,  
$G_{t,(1,g)} = G_{t,g}$ in (\ref{Ftg2}) when choosing $\Gamma_g := -\xi(g^{-1},g) + \gamma_g$ (mod $2\pi$)
 therein for all $t\in \bR$ and $g\in \cG$.
\end{proposition}
\noindent As before  we notice that, varying 
 $(\chi,g)\in U(1) \times \cG$, the class of functions $\{G_{t,(\chi,g)}\}_{t\in \bR}$ is
one-to-one associated with the corresponding transformation acting in the extended phase-space   $G_{(r,g)} : \bR \times \bR \times \bR^{6N}\times \bT^N \times \bR^N  
 \to \bR \times \bR^{6N}\times \bT^N \times \bR^N$:
  \beq
G_{(r,g)} :  \begin{bmatrix} t\\
{\bf x}_k\\  {\bf p}_k\\ \zeta_k, \\ n_k
\end{bmatrix} 
\mapsto
\begin{bmatrix} t'\\
{\bf x}'_k\\  {\bf p}'_k\\ \zeta'_k, \\ n'_k
\end{bmatrix}  := 
\begin{bmatrix} t + c_g\\ {\bf c}_g + t{\bf v}_g + R_g{\bf x}_k \\ 
 |n_k|{\bf v}_g+R_g{\bf p}_k \\ 
 \zeta_k +\mbox{sign}(n_k)(\varphi_{g^{-1}}(t, {\bf x}_k)  - r -\xi(g^{-1},g))\:\: \mbox{mod $2\pi$}\\
  n_k
\end{bmatrix}\:. \label{Grg}
\eeq
The composition rule  $G_{t,(\chi',g')}G_{t,(\chi,g)}
= G_{t,(\chi',g')(\chi,g)}$ is equivalent to  $G_{(\chi',g')}G_{(\chi,g)}
= G_{(\chi',g')(\chi,g)}$.\\

\noindent {\em Proof of Proposition \ref{p12}}. It is essentially identical to that of Proposition \ref{p1} with the obvious changes.
$\Box$
 
 \subsection{Modified quantum system}\label{MQS} We can pass to the quantization of the given system. Now each canonical pair
 $n_k, \zeta_k$ is associated to the canonical pair of densely-defined unbounded self-adjoint 
  operators $\widehat{n}_k, \widehat{\zeta}_k$ defined on 
 $\sH_k := L^2(\bS^1, d\zeta_k)$. The operator $ \widehat{\zeta}_k$ is the standard multiplicative operator with $\zeta$, 
when $\bS^1$ is identified to the segment $[0,2\pi)$, the range of the coordinate $\zeta$.
 The operator $\widehat{n}_k$  is the unique self-adjoint extension of $-i \frac{\partial}{\partial \zeta_k}$ defined on the smooth periodic functions 
 on $\bS^1$ (from Nelson's theorem it follows that the considered symmetric operator is essentially self-adjoint since the exponentials $e^{in_k \zeta}$ define a set of analytic vectors spanning a dense set). 
 Obviously $\sigma(\widehat{n}_k)= \sigma_p(\widehat{n}_k) = \bZ$.
The  whole Hilbert space can be defined as the closed subspace of $L^2(\bR^{3N} \times \bT^N, d^{3N}x \otimes d^N\zeta)$ 
 given by the orthogonal direct sum:
\beq
\sH := \bigoplus_{n_k \in \bZ\setminus\{0\}\:, k=1,\ldots, N} \sH_{\{n_k\}}\:, \quad \sH_{\{n_k\}} \simeq 
L^2(\bR^{3N}, d^{3N}x)\label{HEXT}
\eeq
where $\sH_{\{n_k\}}$ is the common eigenspace of all $\widehat{n}_k$ with the eigenvalues  in the set $\{n_k\}_{k=1,\ldots,N}$.
We have removed all the eigenspaces where at least one of $n_k$ vanishes. As each $\widehat{n}_k$ is a constant of motion, $\sH$ is invariant under time evolution. On the other hand  $\sH$ is invariant under the action of the (centrally extended) Galileian group as we are going to prove. The operator $\sH$ is thus well defined. A total mass self-adjoint operator is well defined with positive point-wise spectrum $\sigma(\widehat{M})= \sigma_p(\widehat{M}) = \{N,N+1,N+2,\ldots\}$ so that no problems arise in stating the superselection rule within the standard form.
 It is (taking into account Theorem VIII.33 in \cite{RS}, when interpreting the right-hand side)
$$\widehat{M} := \overline{\sum_{k=1}^N |n_k|}\:,$$
 where the $\widehat{n}_k$, initially defined on $L^2(\bS, d\zeta_k)$, are supposed to be restricted  to the orthogonal of the kernel in view of our definition of $\sH$.
 
\begin{proposition}\label{p22} Consider the central extension $U(1) \times_\omega \cG$, where $\omega = e^{i\xi}$, $\xi$ being that associated with 
the functions $\varphi_g$ as in (\ref{phases}) determined by an arbitrary  choice of the constants $\gamma_g$ and
focus on the class $\{U_{(\chi,g)}\}_{(\chi,g)\in \bR \times_\omega \cG}$ of transformations on the extended Hilbert space (\ref{HEXT})
$U_{(\chi,g)} : \sH \to \sH$, where $U_{t}(\{n_k\})$ is the time evolutor in the subspace
 $\sH_{\{n_k\}}$
\beq
\left(U_{(\chi,g)} \Psi\right) (\{n_k\}, \{{\bf x}_k\})  = \chi^{\widehat{M}} e^{i({\bf v}_g \cdot {\bf x_j} + \gamma_g)\widehat{M}} U_{-b_g}(\{n_k\})
 \Psi\left(\{n_k\},  \{R_g^{-1}({\bf x}_k -c_g {\bf v}_g -{\bf c}_g)\}\right)\:,
\eeq
They satisfy the following.

(i)  $U(1) \times_\omega \cG \ni (\chi, g) \mapsto U_{(\chi,g)}$ is a unitary representation of $U(1) \times_\omega \cG$ that, for $\chi=1$, reduces to the unitary projective representation (\ref{Ug}) in each subspace $\sH_{\{n_k\}}$,
associated with the multiplier $e^{i M \xi(g',g)}$ with the total  mass $M= \sum_{j=1}^N |n_j|$.

(ii) If $\Psi \in \sH_{\{n'_k\}}$, passing into $\zeta$-picture, defining 
 $\check{\Psi}' (\{\zeta_k\}, \{{\bf x}_k\})  :=  \left(\check{U}_{(\chi,g)} \check{\Psi}\right) (\{\zeta_k\}, \{{\bf x}_k\})$
and taking the time evolution into account, one has 
\beq
\check{\Psi}' (t, \{\zeta_k\},\{{\bf x}_k\})  = \check{\Psi}(\check{G}^{\{n'_k\}}_{(\chi,g)^{-1}}(t, \{\zeta_k\}, \{{\bf x}_k\})) 
\eeq
where $\check{G}^{\{n'_k\}}_{(\chi,g)}$ is the function that considers only the first, the second and the fourth
components  of the function (\ref{Grg}) taking the signs of the values $n'_k$ into account:
 \beq
\check{G}^{\{n'_k\}}_{(\chi,g)} :  \begin{bmatrix} t\\
{\bf x}_k\\  {\bf p}_k\\ \zeta_k, \\ n'_k
\end{bmatrix} 
\mapsto
\begin{bmatrix} t + c_g\\ 
  {\bf c}_g + t{\bf v}_g + R_g{\bf x}_k\\ 
 \zeta_k + \mbox{\em sign}(n'_k)(\varphi_{g^{-1}}(t, {\bf x}_k)  - r -\xi(g^{-1},g))
\end{bmatrix}\:. 
\eeq
In particular $\Psi'$ still belongs to $\sH_{\{n'_k\}}$ at any time $t$.
\end{proposition}

\begin{proof} The fact that $\Psi'$ still belongs to $\sH_{\{n'_k\}}$ at any time $t$ if $\Psi$ does 
immediately arises from the definition of $U_{(\chi,g)}$ since it does not changes the values of the $n_k$.
The first statement in the thesis easily arises from the definition $\sH$, taking into account (a) that  $U_{(r,g)}$ leaves fixed every $\sH_{\{n_k\}}$
and it  reduces to unitary operator $e^{i M r}U_g^{(M)}$
on every $\sH_{\{n_k\}}$ with the total  mass $M= \sum_{j=1}^N |n|_j$, and (b) that (\ref{proj}) 
holds therein. Concerning (ii), following the same route to reach the identity (\ref{popular}) we have, if $\chi=e^{ir}$
$$\Psi' (\{n_k\}, t, \{{\bf x}_k\})
= e^{i\sum_k |n_k|(\varphi (t, {\bf x}_k) + r)}\Psi (\{n_k\}, g^{-1}(t,\{{\bf x}_k\})) \:,$$
where $t$
(whose position is temporarily chosen as follows for sake of convenience)
indicates the time evolution.
Therefore:
$$\check{\Psi}'(\{\zeta_k\}, t, {\bf x}_k)
= \frac{1}{(2\pi)^{N/2}}{e^{i\sum_k n_k(\zeta_k +\mbox{sign}(n'_k)(\varphi_g (t, {\bf x}_k) + r))}}\Psi (\{n'_k\}, g^{-1}(t,\{{\bf x}_k\}))$$
$$= \check{\Psi}(\{\zeta_k +\mbox{sign}(n'_k)(\varphi_g (t, {\bf x}_k) + r)\},  g^{-1}(t,\{{\bf x}_k\}))\:.$$
The whole argument of $\check{\Psi}$ in the right-hand side, restoring the standard position of the variable $t$, can be written
$\check{G}^{\{n'_k\}}_{(r,g)^{-1}}(t, \{\zeta_k\}, \{{\bf x}_k\})$, paying attention to the identities
$(r,g)^{-1} = (-r - \xi(g^{-1},g), g^{-1})$ and $\xi(g^{-1},g)= \xi(g, g^{-1})$.
\end{proof}

\noindent {\bf Remarks}.\\ 
{\bf (1)} Obviously, the values of the masses can be changed just by modifying the radius of each $\bS^1$.\\
{\bf (2)} The presented model transforms the physically impossible sign of the mass into a new degree of freedom. Each particle has an internal space that can be used to describe some charge taking values in $\bZ \setminus \{0\}$. Indeed, for $N=1$ the Hilbert space of the system can be re-organised as follows:
$\sH= \oplus_{r=1,2,\ldots} (\sH_r\oplus \sH_{-r})$ where
$\sH_r\oplus \sH_{-r}\simeq  \bC^2 \otimes L^2(\bR^3, d^3x)$. In this approach $\widehat{n}$ is a charge operator.
The extension to the case of $N>1$ is obvious. Notice that the total charge defined in that way turns out to be automatically conserved under time evolution, since every $\widehat{n}_k$ commutes with the total  Hamiltonian operator.\\
{\bf (3)} The model admits an anti-unitary (since $\widehat{H}$ is bounded below) time-reversal operator $T$,
 thus verifying $TU_t= U_{-t}T$ if $U_t$ is the time evolution operator. It is, the bar denoting the complex conjugation:
$$(\check{T}\check{\Psi}) (\{\zeta_k\}, \{{\bf x}_k\}) := \overline{\check{\Psi}(\{\zeta_k\}, \{{\bf x}_k\})}$$ 
where we are working in $\zeta$ picture so that $\check{\Psi} \in L^2(\bT^N\times \bR^{3N}, d^N\zeta \otimes d^{3N}x)$. 
Actually, since the variables $\zeta_k$ are, at the end,  unobservable and they do  not appear in the Hamiltonian
operator, another possible time reversal operator may be (remembering that $\zeta_k$ is defned mod $2\pi$):
$$(\check{T}'\check{\Psi}) (\{\zeta_k\}, \{{\bf x}_k\}) := \overline{\check{\Psi}(\{-\zeta_k\}, \{{\bf x}_k\})}\:.$$ 
Reversing the sign of $\zeta_k$ is equivalent to reversing that of $n_k$.
In this context, assuming the viewpoint in (2), where the $\widehat{n}_k$ are charge operators, one may define a unitary  $k$-charge 
conjugation operator $C_k$ such that $C_kC_k=I$ and  $C_kU_t= U_tC_k$:
 $$(C_h{\Psi}) (\{n_k\}, \{{\bf x}_k\}) :={\Psi}(\{\eta_k^{(h)}n_k\}, \{{\bf x}_k\})\:,\quad \mbox{where $\eta_h^{(h)}=-1$ 
 and $\eta_k^{(h)}=1$ if $k\neq h$.}$$
 With this definition, $T= T' C_1\cdots C_N$. Within this approach (loosely following a procedure similar to that exploited in chromodynamics to construct uncolored states)
 uncharged systems could be described by the states in the subspace of $\sH$ verifying $C_k\Psi=
\Psi$ for every $k=1,\ldots,N$, that is $\Psi(\{n_k\}, \{{\bf x}_k\})={\Psi}(\{- n_k\}, \{{\bf x}_k\})$ for all values of the $n_k$. Therefore, for fixed masses 
$\{m_k\}$, the information on the state is contained in a 
unique wavefunction $\Psi(\{{\bf x}_k\})_{\{m_k\}} = {\Psi}(\{\pm m_k\}, \{{\bf x}_k\})$ and the internal degrees of freedom 
given by the signs of the $n_k$ do not play any role.\\
 {\bf (4)} The appearance of the internal degree of freedom, interpreted as a charge, to get rid of the problem of the sign of the mass, may seem annoying.
 Actually the one we have exploited is not the only possibility to define a positive  mass operator with point-wise spectrum. The core of the problem lies in the quantization procedure based on Dirac's approach, consisting in associating 
 self-adjoint operators with classical observables, with commutation rules preserved. 
 Keeping the constraint on commutation relations,
 one may drop the request of self-adjointness and stick to maximally symmetric operators, spectrally decomposed in terms of POVM instead 
 PVM \cite{BGL}. In this respect, the couple of classical variables $\zeta_k, m_k$ 
 can be quantized into a pair made of a symmetric operator $\widehat{\zeta}_k$
 whose spectral measure (POVM) covers $\bR$ completely,
 and a self-adjoint operator $\widehat{m}_k$ whose spectrum is $\bN$, all that maintaining the commutation relations on a dense subspace. 
 That is nothing but the standard procedure to quantize the pair 
 phase number-of-particle \cite{BGL}. Almost all the features of our model should be preserved following that new way.
 
\section{Mass superselection rule as a dynamical process for more general non-relativistic physical systems with point-wise mass-operator spectrum} \label{secquasifine}
In the previous section we have explicitly constructed a particular model where the mass observable, the (unbounded) self-adjoint operator $\widehat{M}$,
 is defined by a natural quantization procedure and $\sigma(\widehat{M}) = \sigma_p(\widehat{M}) = \{N, N+1, N+2, \ldots\}$.
In this section we assume, more generally, that a separable Hilbert space $\sH$, describing some physical system,  can be decomposed with respect to the spectral measure of  a (unbounded in general)
self-adjoint operator with pure point spectrum
$\widehat{M}$, not necessarily the previous one.  
\beq \sH = \bigoplus_{M \in \sigma_p(\widehat{M})}  \sH_M\:, \quad \mbox{where $\sigma(\widehat{M}) = \sigma_p(\widehat{M}) \subset (0,+\infty)$}\label{Hdec}\eeq
 In the following $P_M$ denotes the orthogonal projector on $\sH_M$ and $\gB(\sH)$ the $C^*$-algebra of bounded operators $B: \sH \to \sH$. 
Moreover, we denote by $\gS(\sH)$ the convex body of the positive trace-class  operators $\rho \in \gB(\sH)$ with $tr(\rho)=1$, and by 
$\gS_B$ the convex subset of those   whose eigenvectors are eigenvectors of $\widehat{M}$ too. In other words:
\beq \gS_B = \{\rho \in \gS(\sH)\:|\: P_M\rho = \rho P_M\quad \mbox{for all $M\in \sigma(\widehat{M})$}\}\:. \label{S}\eeq
Finally we define the von Neumann algebra in $\gB(\sH)$, whose center is non-trivial:
 \beq \gA_B:= \{P_M\:|\: M \in \sigma(\widehat{M}) \}'\label{A}\:.\eeq
 The supreselection rule for $\widehat{M}$ (in its elementary form in Hilbert spaces) declares  that only elements $\rho \in \gS_B$ 
 make sense as (normal) states on the physical system. However there is a subsequent 
 dual  restriction concerning the permitted observables.
Indeed, adopting the standard von Neumann-L\"uders'  postulate on the reduction of the state, to prevent
 the production of forbidden states as outcomes of  measurement procedures, one is committed to assume that only the (self-adjoint) elements $A \in \gA_B$ are allowed 
 as bounded observables. Actually both requirements can be relaxed in a sense that will turn out to be useful 
 in providing, following Giulini's ideas, a dynamical interpretation of the superselection rule itself. 

\subsection{How forbidden states determine permitted states and 
forbidden observables determine allowed observables} \label{secH}
Within this section we focus on the interplay of permitted and forbidden states in view of the superselection rule and the corresponding dual interplay 
at the level of the observables. This interplay will be subsequently exploited to grasp a dynamical interpretation of Bargmann superselection rule.\\
As is known, a forbidden coherent combination $\Psi\in \sH$ of mass-defined states can represent a 
permitted state anyway: A mixture of allowed pure states. This result generalizes to mixed states as we shall see shortly, for the moment we stick to the simplest situation with the following lemma.
The proofs of all propositions in this section can be found in the appendix.

\begin{lemma}\label{propstates} Referring to the Hilbertian decomposition (\ref{Hdec})  of the Hilbert space $\sH$ and the von Neumann algebra 
 $\gA_B$ in (\ref{A}) and the class of states $\gS_B$ in (\ref{S}),
consider  a vector $\Psi = \sum_{M\in S} \Psi_M$ with   
$S\subset \sigma(\widehat{M})$,  $\Psi_M \in \sH_M \setminus\{0\}$ and $||\Psi||=1$. For every $A\in \gA$ it holds:
 \beq
 \langle \Psi| A \Psi \rangle = \sum_{M\in S} \langle \Psi_M |A \Psi_M\rangle = tr(\rho_\Psi A)\:, \label{lemma}
 \eeq 
 where the positive, trace class operator with unitary trace $\rho_\Psi\in \gS_B$ is defined as:
\beq \rho_\Psi :=   \sum_{M\in S} p_M \Phi_M \langle \Phi_M |\: \cdot \:\rangle \quad \mbox{with
 $\Phi_M := \frac{\Psi_M}{||\Psi_M||}$ and $p_M := ||\Psi_M||^2$,}
 \eeq
 the convergence of the series being in the uniform topology.\\
 The map $\sH \ni \Psi \mapsto \rho_\Psi$ is not injective, since $\Psi' = \sum_{M\in S} \chi_M \Psi_M$ produces the same $\rho_\Psi$
for every choice of the phases $\chi_M \in U(1)$.
\end{lemma}

\noindent 
  A forbidden observable  $A \in \gB(A)$  can represent a permitted one  $A_B \in \gA_B$:
A ``convex'' linear combination of permitted observables.
\begin{theorem} \label{propobs} Referring to the Hilbertian decomposition (\ref{Hdec})  of the Hilbert space $\sH$ and the von Neumann algebra 
 $\gA_B$ in (\ref{A}) and the class of states $\gS_B$ in (\ref{S}),
consider  an operator  $A \in \gB(\sH)$. For every $\rho \in \gS_B$:
 \beq
tr(\rho A)    = tr(\rho A_B) \:,\label{tr}
 \eeq 
 where the operator $A_B \in \gA_B$ is defined as (where the convergence of the series is in the strong operator topology):  
\beq A_B :=   \sum_{M\in \sigma(\widehat{M})} P_M AP_M\:. \label{BA}\eeq
 The map $\gB(\sH) \ni A \mapsto A_B \in \gA_B$ is linear and satisfies the further properties:\\
  (i) $A_B=A$ if $A\in \gA$ so that $\gB(\sH) \ni A \mapsto A_B \in \gA_B$ is surjective, \\
 (ii) $||A_B|| \leq ||A||$,  \\
 (iii) $(A^*)_B=(A_B)^*$,\\ 
  (iv) If $P$ 
 is an orthogonal projector, $P_B$ is an {\bf effect} {\em \cite{BGL}}: a positive element in $\gB(\sH)$ bounded above by $I$.\\
 (v) $A_B = A'_B$ if and only if $tr(\rho A)= tr(\rho A')$ for all $\rho \in \gS_B$.
\end{theorem}

\noindent The map $\gB(\sH) \ni A \mapsto A_B \in \gA_B$ is not injective as one can easily prove.
Collecting both results we have a useful corollary generalizing Lemma \ref{propstates} to mixtures as pre-announced.

\begin{theorem}\label{teotot} With the hypotheses of Theorem \ref{propobs}, 
consider  trace-class positive operator $\rho$ with  $tr(\rho)=1$.  It turns out that $\rho_B \in \gS_B$ and that, 
for every $A\in \gA_B$:
 \beq
 tr(\rho A) = tr(\rho_B A) \:. \label{trtr}
 \eeq 
 The following facts hold.\\
(i) $\rho_B=\rho'_B$ if and only if $tr(\rho_B A)=tr(\rho'_B A)$ for all $A\in \gA_B$.\\
 (ii) If $\rho = \Psi \langle \Psi| \:\cdot \: \rangle$ then $\rho_B = \rho_\Psi$ as il Lemma \ref{propstates}.
\end{theorem} 
 
\subsection{Allowed states and observables as equivalence classes of general states and observables} \label{secclasses}
 Summing up, we can say that, in the presence of Bargmann superselection rule, the set of states and that of  observables separately 
 decompose into the disjoint union of equivalence classes. Two states $\rho,\rho'\in \gS(\sH)$ are equivalent if they determine the same permitted state 
 $\rho_B=\rho'_B\in \gS_B$
 or, that is the same, if they produce the same results when applied to the same physically admissible observable $A\in \gA_B$.
 Similarly, two observables $A,A'\in \gB(\sH)$ are equivalent if they determine the same permitted observable $A_B=A'_B \in \gA_B$
 or, that is the same, if they produce the same results when applied to the same physically admissible state $\rho \in \gS_B$.
 It seems that, assuming that general states and observables make physical sense anyway, the existence of the superselection rule prevents us  to
  directly handle observables and states: we can just handle equivalence classes of these objects. If this obstruction is due to some physical phenomenon, the 
only   place where it can work is during the measurement procedures.
If $A\in \gB(\sH)$
and $\rho \in \gS(\sH)$, in experiments, we suppose to be evaluating  $tr(\rho A)$, while actually we are evaluating something 
different  and the final numerical outcome coincides to $tr(A_B \rho_B)$ instead of $tr(\rho A)$.  
For that reason we shall indicate by $\langle \langle A\rangle \rangle_\rho$  the physical measurement process of $A$ with respect to 
$\rho$ in the presence of the superselection rule and we wish to construct a mathematical model of that measurement process.

  \subsection{Superselection as averaging procedure on $U(1)$}  
   Referring to  the physical system discussed in Sec.\ref{MQS}, when we deal with the observables allowed by the superselection rule, we cannot experimentally observe the action of  $U(1)$ (appearing in the  central extension  $U(1) \times_\omega \cG$). Even if it does not suggest yet any physical process giving rise to the superselection rule,
   this remark can be 
  promoted to another formulation of the superselection rule itself.  As before, we assume that: even observables   and states in principle forbidden by the superselection rule can initially be considered, that
the true invariance group is $U(1)\times_\omega \cG$,  
  and that the 
  superselection rule enters physics just during measurements. As we cannot experimentally observe the action of  $U(1)$ looking at the results of the measurements,   the observable $U^*_{(e^{ir}, e)} A U_{(e^{ir}, e)}$ must be, for every value of $r\in \bR$,
   physically  indistinguishable from  $A$ when we perform measurements. Computing  $\langle \langle A\rangle \rangle_\rho$,  the final numerical outcome 
   has to coincide to $tr(\rho_BA_B)$.
   This suggests to suppose that real measurements procedure $\langle \langle A\rangle \rangle_\rho$  also includes an averaging procedure over  the 
   group $U(1)$ with respect to its Haar measure, since Haar measure is invariant under the action of the group itself. 
\beq \langle \langle A\rangle \rangle_\rho := \frac{1}{2\pi}\int_{0}^{2\pi}  tr\left( \rho U^*_{(e^{i\theta}, e)} A U_{(e^{i\theta}, e)} \right) \: d\theta\:.\label{avint}\eeq  
With this definition it results indeed $\langle \langle U^*_{(e^{ir}, e)} A U_{(e^{ir}, e)}\rangle \rangle_\rho = \langle \langle A \rangle \rangle_\rho$
for every $r \in \bR$.
   The integral average in (\ref{avint}) can be passed on the state taking advantage on the cyclic property of the trace.\\
   The following theorem proves that our idea grasps, in fact, some insight towards a dynamical model of the superselection rule, since the averaging procedure not only follows from the superselection rule, but even it
    gives rise to the superselection rule itself, in the sense that $\langle \langle A\rangle \rangle_\rho=
    tr(\rho_BA_B)$.

\begin{theorem} 
Consider the physical system discussed in Sec.\ref{MQS} where, in particular $\sigma(\widehat{M})= \{N, N+1, \ldots\}$ and the separable Hilbert space $\sH$ is decomposed as in (\ref{Hdec}). 
 If $A\in \gB(\sH)$ and  $\rho\in \gS(\sH)$, then:
 \beq
 \frac{1}{2\pi}\int_{0}^{2\pi}  tr\left( \rho U^*_{(e^{i\theta}, e)} A  U_{(e^{i\theta}, e)} \right)\: d\theta
 =
 \frac{1}{2\pi}\int_{0}^{2\pi}  tr\left( U_{(e^{i\theta}, e)} \rho U^*_{(e^{i\theta}, e)} A  \right) \: d\theta =
  tr(\rho_B A_B)\:.\label{avrho}
 \eeq
 and in particular, if $\Psi \in \sH$ with $||\Psi||=1$:
 \beq
 \frac{1}{2\pi} \int_0^{2\pi} \langle \Psi| U^*_{(e^{i\theta}, e)} A  U_{(e^{i\theta}, e)}  \Psi \rangle\: d\theta
 = tr(\rho_\Psi A_B)\:.\label{avPsi}
 \eeq
\end{theorem}  
 
  \begin{proof} We start by proving (\ref{avPsi}).
  To this end consider the map
  $$\sH \times \sH \ni (x,y) \mapsto  J(x,y) :=\frac{1}{2\pi}\int_0^{2\pi} \langle x| U^*_{(e^{i\theta}, e)} A  U_{(e^{i\theta}, e)} y \rangle d\theta\:.$$
  It is well-defined because the integrand function is continuous, since $\theta \mapsto U_{(e^{i\theta}, e)}
  = e^{i\theta \widehat{M}}$ is strongly continuous and the scalar product is jointly continuous in its arguments.  $(x,y) \mapsto  J(x,y)$ is linear in $y$ and anti linear in $x$. Moreover, essentially taking advantage from the Cauchy-Schwarz inequality and the fact that 
  $||U_{(e^{i\theta}, e)}||= ||U^*_{(e^{i\theta}, e)}||=1$, one also has 
  $|J(x,y)| \leq ||A|| ||x||\: ||y||$. By a straightforward use of Riesz theorem, one concludes that there exists
  $J\in \gB(\sH)$ with $||J|| \leq ||A||$ and $\langle x|Jy\rangle = J(x,y)$. We claim that $J=A_B$.
  To prove it, it is enough establishing that $\langle x|Jy\rangle = \langle x| A_B y\rangle$ for $x,y \in R$,
  $R$ being a dense subset of $\sH$. Consider $$R :=\left\{ \left.\Psi = \sum_{M \in S} \Psi_M \:\right|\quad 
  \mbox{$S\subset \sigma(\widehat{M})$, $S$ finite}\right\}\:.$$
  $R$ is dense because of the orthogonal decomposition (\ref{Hdec}). Moreover, if $\Psi,\Psi' \in R$, for some sufficiently large $K<+\infty$
$$\langle \Psi| J\Psi' \rangle =\int_0^{2\pi} \sum_{M,M'<K} 
   \frac{e^{i(M'-M)\theta}}{2\pi} \langle \Psi_M | A \Psi'_{M'}\rangle d\theta  =  \sum_{M<K}  
  \langle \Psi_M | A \Psi'_M\rangle  =  \sum_{M<K}  
  \langle \Psi | P_MAP_M \Psi'\rangle $$ $$ =
  \langle \Psi| A_B \Psi' \rangle\:.$$  
  Thus $J=A_B$ so that, in particular, for every $\Psi, \Psi' \in \sH$:
\beq  
\frac{1}{2\pi} \int_0^{2\pi} \langle \Psi| U^*_{(e^{i\theta}, e)} A  U_{(e^{i\theta}, e)}  \Psi' \rangle\: d\theta =
 \langle \Psi| J \Psi'\rangle =  \langle \Psi| A_B \Psi' \rangle
 \eeq
  and, for $\Psi=\Psi'$:
  $$\frac{1}{2\pi} \int_0^{2\pi} \langle \Psi| U^*_{(e^{i\theta}, e)} A  U_{(e^{i\theta}, e)}  \Psi \rangle\: d\theta =  \langle \Psi| A_B \Psi \rangle$$
  that implies (\ref{avPsi}) due to (\ref{lemma}) since $A_B \in \gA_B$.\\
  Let us pass to the proof of (\ref{avrho}).
It is sufficient to prove the former identity, as then the latter arises from the cyclic property of the trace.  
   First of all, we notice that 
  the function  in the integrand in the left-hand side in (\ref{avrho}), $\theta \mapsto tr\left( \rho U^*_{(e^{i\theta}, e)} A  U_{(e^{i\theta}, e)} \right)$,
  is measurable as it is the sum of a series of continuous functions (expanding the trace on a basis of eigenvectors of $\rho$), so the integration  
  makes sense provided the function is integrable. We will compute the trace with respect to a Hilbert basis $\{\psi_n\}_{n\in \bN}$ (labelled on $\bN$ since 
  $\sH$ is separable in the considered case) of eigenvectors of $\rho$.  So that $\rho = \sum_{n\in \bN} p_n\psi_n \langle \psi_n | \:\cdot \:\rangle$ where
  $0\leq p_n \leq 1$ and $\sum_n p_n =1$.
   Let us prove that $\theta \mapsto tr\left( \rho U^*_{(e^{i\theta}, e)} A  U_{(e^{i\theta}, e)} \right)$ is integrable  exploiting Fubini-Tonelli with respect to the product measure 
  given by the product of $d\theta$ and the measure counting the points on $\bN$.  We have, as $||U_{(e^{i\theta}, e)} \psi_n||= ||\psi_n||=1$:
$$|\langle \psi_n | \rho U^*_{(e^{i\theta}, e)} A  U_{(e^{i\theta}, e)} \psi_n \rangle|  = p_n |\langle \psi_n | U^*_{(e^{i\theta}, e)} A  U_{(e^{i\theta}, e)} \psi_n \rangle| 
= p_n \langle U_{(e^{i\theta}, e)} \psi_n |  A  U_{(e^{i\theta}, e)} \psi_n \rangle  \leq p_n ||A||\:.$$  
  So
  $$\int_0^{2\pi} \sum_{n\in \bN} |\langle \psi_n | \rho U^*_{(e^{i\theta}, e)} A  U_{(e^{i\theta}, e)} \psi_n \rangle| d\theta \leq ||A||
  2\pi \sum_m p_n = 2\pi \:||A||\:.$$
  Fubini-Tonelli theorem implies that $\theta \mapsto tr\left( \rho U^*_{(e^{i\theta}, e)} A  U_{(e^{i\theta}, e)} \right)$ is integrable and we can swap the symbol of series 
  and that of integral:
  $$ \frac{1}{2\pi}\int_0^{2\pi} tr(\rho U^*_{(e^{i\theta}, e)} A  U_{(e^{i\theta}, e)}) d\theta = 
  \frac{1}{2\pi}\int_0^{2\pi} \sum_{n\in \bN} p_n \langle \psi_n |U^*_{(e^{i\theta}, e)} A  U_{(e^{i\theta}, e)} \psi_n \rangle d\theta$$
 $$ =  \sum_{n\in \bN} p_n \frac{1}{2\pi} \int_0^{2\pi}\langle \psi_n |U^*_{(e^{i\theta}, e)} A  U_{(e^{i\theta}, e)} \psi_n \rangle d\theta =
  \sum_{n\in \bN}   p_n tr(\rho_{\psi_n} A_B) =    \sum_{n\in \bN}   p_n \langle \psi_n| A_B \psi_n \rangle$$
  $$ = tr (\rho A_B) = tr (\rho_B A_B)\:.$$
  Above we have exploited (\ref{avPsi}), Lemma (\ref{propstates}) and Theorem (\ref{propobs}).
  \end{proof}

  \subsection{Superselection as an effective  dynamical process}
 As we have seen, the appearance of the superselection rule can be interpreted as an averaging procedure over the unobservable 
 part of the extended  Galileian group (supposed to be the true symmetry group of the system) when assuming that also  states and observables, in principle forbidden, can actually be handled. However that interpretation  works for the toy model discussed in Sec.\ref{MQS} 
that has a quite particular spectrum for the mass operator $\widehat{M}$: The difference of eigenvalues has to be an integer number. 
Otherwise one has to change the interval of integration in $\theta$ to match the period of the imaginary exponential and, in general, one can hardly fix it to agree with all the differences of the mass eigenvalues in a realistic model, since these differences may have irrational ratio. However the procedure can be improved taking the period larger and larger and exploiting the same idea as in Riemann-Lebesgue's lemma.\\
All that suggests that there is another, much more realistic,  interpretation of the one-parameter group 
 $\bR \ni \theta \mapsto e^{i \theta \widehat{M}}$, that permits to give a truly  dynamical interpretation of the 
 mass superselection rule just realizing Giulini's proposal. One has to pay attention to the fact that 
 standard quantum mechanics is an approximation of relativistic quantum mechanics (i.e. relativistic quantum field theory described in the one-particle space). Usually that approximation is performed at classical (not quantum) level and afterwards the model is quantized. At classical level the mass is a number and thus the  constant $Mc^2$, added to the classical Hamiltonian when performing the simplest 
 non-relativistic approximation of the (kinetic) energy,
$$\sqrt{c^2{\bf p}^2 + m^2c^4} = mc^2 + \frac{{\bf p}^2}{2m} + mc^2 \:O({\bf p}^4/m^4c^4)\:,$$ 
is usually neglected. 
However,  the lesson learned from the models discussed in \cite{Giulini} and in this paper 
 suggests instead that the mass has to be considered a full-fledged operator, that cannot be trivially dropped. Thus the term $\widehat{M}c^2$  has to be included in the potential $\widehat{V}$ appearing in the definition
 of the Hamiltonian operator (\ref{spaceH}). It accounts for the sum of the contributions $\widehat{m}_kc^2$ due to all the particles forming the system. In this way, no problems arise with the invariance under any central extension $U(1) \times_\omega \cG$, provided the representation of the $U(1)$ group in the center is given by
 $U(1) \ni \chi \mapsto \chi^{\widehat{M}}$. In particular, it holds for the model constructed  in Sec.\ref{MQS}
 where the time reversal symmetry is preserved. 
The classical Hamiltonian can be re-written as $\overline{\widehat{M}c^2 + \widehat{H}}$, where $\widehat{H}$ takes the form (\ref{spaceH}), where $\widehat{V}$ does not contain the term $\widehat{M}c^2$.
 In (\ref{spaceH}), the masses $m_k$ are, as before, replaced with corresponding operators $\widehat{m}_k$ commuting with all the usual dynamical variables of the system. In particular the spectral measure of $\widehat{M}$ commutes with that of $\widehat{H}$. We do  not assume any specific form for the operator $\widehat{M}$ but only that $\sigma(\widehat{M}) = \sigma_p(\widehat{M})\subset (0,+\infty)$ so that the decomposition (\ref{Hdec}) holds\footnote{We could assume, more weakly, that $\emptyset \neq \sigma_p(\widehat{M})\subset (0,+\infty)$. So that only a closed subspace $\sH_p$ of $\sH$
 is spanned by the direct sum of eigenspace of $\widehat{M}$. In this case, we  replace $\sH$ for $\sH_p$ in the rest of the section, preserving the results.}. 
 The time evolution operator decomposes as, restoring physical units
\beq U_{t} = e^{-i \frac{tc^2}{\hbar}\widehat{M}} e^{- i\frac{t}{\hbar}\widehat{H}} \:.\label{Ut}\eeq
{\em In each eigenspace of the mass operator, the second evolution operator in the right-hand side  coincides with that of the standard formulation of non-relativistic quantum mechanics}.\\
Let us compute the expectation value of an observable $A\in \gB(\sH)$ with respect to 
a state $\Psi \in \sH$, both violating the superselection rule.
We assume that, if $D(\widehat{H})$  is the dense domain of $\widehat{H}$,
\beq \Psi = \sum_{M\in S} \Psi_M\:,\quad  \mbox{$||\Psi|| =1$, $\Psi_M \in D(\widehat{H})$,  $S$ finite $S\subset \sigma_p(\widehat{M})$.} \label{PSI}\eeq
 We suppose that, in real measurement processes, physical instruments 
 perform a temporal averaging along a period of time $T$ depending on the interaction system-instruments. So we define the effective expectation value
 of $A$ in the state represented by $\Psi$:
\beq \langle\langle A \rangle\rangle_{\Psi,T} := \frac{1}{T}\int_0^T \langle \Psi | U^*_t A U_t \Psi\rangle dt\:.\label{W}\eeq
{\bf Remark}.  In the following,  if  $O$ is a self-adjoint, not necessarily bounded, operator with domain $D(O)$, and $\Psi \in D(O)$:
$$\langle O\rangle_{\Psi} := \int_{\sigma(O)} \lambda\: d\langle \Psi| P^{(O)}(\lambda) \Psi\rangle\:, \quad 
\Delta O_{\Psi} := \sqrt{\int_{\sigma(O)} (\lambda^2 - \langle O\rangle^2_{\Psi})\: d\langle \Psi| P^{(O)}(\lambda) \Psi\rangle }\:.$$
 Moreover, in the next pair of propositions we adopt the following physically suggestive notation:
 $b<< a$ means $b/a \to 0^+$ and $A\sim B$ means $A-B \to 0$.\\

\noindent We are in a position to state and prove our first result. 
\begin{proposition} \label{PE1}{\em ({\bf Effective superselction rule I})}
Assume that, in the Hilbert space $\sH$,  the operator of time evolution of a quantum system is the product of two commuting strongly 
continuous one-parameter  unitary groups as in  (\ref{Ut}), where 
 $\sigma(\widehat{M}) = \sigma_p(\widehat{M})\subset (0,+\infty)$, and take $\Psi$ as in 
 (\ref{PSI}).\\
 Let $\delta M_\Psi>0$ be the smallest difference of two different eigenvalues  of $\widehat{M}$  in $S$,
 let $N_\Psi$ indicate the number of elements of $S$, and let $E_\Psi$ be a measure 
of the maximal energetic content of $\Psi$ referring to $\widehat{H}$:
\beq
E_\Psi := \max_{M\in S}\left\{\sqrt{\langle \widehat{H}\rangle^2_{\Psi_M} + (\Delta \widehat{H}_{\Psi_M}})^2 \right\}\:.
\eeq
For every $A\in \gB(\sH)$ with the definition (\ref{W})  for $T>0$ and referring to notations and definitions in Sec. {\em \ref{secclasses}} and {\em \ref{secH}}, 
a dynamical implementation of the mass superselection rule arises
\beq
\langle\langle A \rangle\rangle_{\Psi,T} \sim   \sum_{M\in S} \langle \Psi_M |A\Psi_M \rangle = 
tr(\rho_\Psi A_B)\quad \mbox{as}\quad \frac{2\hbar  N_\Psi^2 ||A||}{ \delta M_\Psi c^2} << T << 
\frac{\hbar}{2 N_\Psi^2 ||A|| E_\Psi}\label{sim}\:.
\eeq
 \end{proposition}
\begin{proof} To prove (\ref{sim}), we start noticing that, defining the strongly continuous one-parameter group of unitary operators
$V_t:= e^{-i\frac{t}{\hbar}\widehat{H}}$, we have:
$$ \langle \Psi_M | U^*_t A U_t \Psi_{M'}\rangle =  \langle U_t\Psi_M |  A U_t \Psi_{M'}\rangle
= e^{i\frac{(M-M')c^2t}{\hbar}}\langle V_t\Psi_M |  A V_t \Psi_{M'}\rangle$$
$$= e^{i\frac{(M-M')c^2t}{\hbar}} \left[\langle \Psi_M |  A \Psi_{M'}\rangle +\frac{it}{\hbar} \left( \langle V_{t'}\widehat{H}\Psi_M |  A V_{t'} \Psi_{M'}\rangle- \langle V_{t'}\Psi_M |  A V_{t'}\widehat{H} \Psi_{M'}\rangle\right)\right]$$
where we have used Stone theorem, Lagrange theorem and $t' \in (0,t)$ is an undetermined point.
Therefore we have that:
\beq \langle\langle A \rangle\rangle_{\Psi,T}  = tr(\rho_\Psi A_B) + \hbar \sum_{M \neq M'} \frac{e^{i\frac{(M-M')c^2T}{\hbar}}-1}{i(M'-M)c^2T}
\langle \Psi_M |  A \Psi_{M'}\rangle + R_T\label{expansion}\:.\eeq
To conclude it is enough proving that the second and the third therm in the right hand side can be neglected when, varying 
the various parameters,
$\frac{2\hbar  N_\Psi^2 ||A||}{ \delta M_\Psi c^2 T} \to 0^+$ 
and $\frac{2 N_\Psi^2 ||A|| E_\Psi T}{\hbar} \to 0^+$.
Let us focus on the second term. As $||\Psi||=1$, we have $||\Psi_M|| \leq 1$ and thus:
\beq \left|\hbar \sum_{M \neq M'} \frac{e^{i\frac{(M'-M)c^2T}{\hbar}}-1}{i(M'-M)c^2T}
\langle \Psi_M |  A \Psi_{M'}\rangle\right| \leq \frac{2\hbar N^2_\Psi}{\delta M_\Psi c^2T} ||A||\:. \label{one}\eeq
The rest can be worked out as follows.
\beq R_T= \frac{1}{T}\sum_{M,M'} \int_0^T  
 e^{i\frac{(M'-M)c^2t}{\hbar}} \left[\frac{it}{\hbar} \left( \langle V_{t'}\widehat{H}\Psi_M |  A V_{t'} \Psi_{M'}\rangle- \langle V_{t'}\Psi_M |  A V_{t'}\widehat{H} \Psi_{M'}\rangle\right)\right] dt
\:.
\eeq
 Thus, also using the fact that $V_t'$ is unitary and then preserves the norms,
 $$|R_T| \leq \sum_{M,M'}\frac{1}{\hbar T} \int_0^T  T \left( |\langle V_{t'}\widehat{H}\Psi_M |  A V_{t'} \Psi_{M'}\rangle| +|\langle V_{t'}\Psi_M |  A V_{t'}\widehat{H} \Psi_{M'}\rangle|\right)dt$$
 $$\leq \sum_{M,M'}\frac{T}{\hbar} (||\widehat{H}\Psi_M|| \: ||A|| \: ||\Psi_{M'}|| + ||\widehat{H}\Psi_{M'}|| \: ||A|| \: ||\Psi_{M}|| )$$
$$ \leq \sum_{M,M'}\frac{T}{\hbar}||A||\:  ||\Psi_{M}||\: ||\Psi_{M'}|| \left(\sqrt{\langle \widehat{H}\rangle^2_{\Psi_M} + (\Delta \widehat{H}_{\Psi_M}})^2 
 + \sqrt{\langle \widehat{H}\rangle^2_{\Psi_{M'}} + (\Delta \widehat{H}_{\Psi_{M'}}})^2\right)$$
$$\leq\frac{T}{\hbar}||A||\:   \sum_{M,M'} \left(\sqrt{\langle \widehat{H}\rangle^2_{\Psi_M} + (\Delta \widehat{H}_{\Psi_M}})^2 
 + \sqrt{\langle \widehat{H}\rangle^2_{\Psi_{M'}} + (\Delta \widehat{H}_{\Psi_{M'}}})^2\right)  \leq ||A||\frac{2T N_\Psi^2 E_\Psi}{\hbar}  \:, $$
  where $N_\Psi$ is the number of elements of $S$ and 
  we have used the fact that $||\Psi_M|| \leq 1$ because $||\Psi||=1$ and that, using the spectral measure of  $\widehat{H}$,
 $||\widehat{H}\Psi_{M}||^2 = \langle \widehat{H}\rangle^2_{\Psi_{M}} + (\Delta \widehat{H}_{\Psi_{M}})^2$.
 From the obtained estimate and  (\ref{one}) it is clear that $\frac{2\hbar N_\Psi^2 ||A||}{ \delta M_\Psi c^2} <<  \:T << \frac{\hbar}{2 ||A|| N_\Psi^2 E_\Psi}$,
that is $\frac{2\hbar  N_\Psi^2 ||A||}{ \delta M_\Psi c^2 T} \to 0^+$ 
and $\frac{2 N_\Psi^2 ||A|| E_\Psi T}{\hbar} \to 0^+$, 
 imply the first identity in  (\ref{sim}). 
 \end{proof}
 
\noindent {\bf Remark}. Our result can easily  be extended to finite mixtures of 
states $\Psi$, each containing a finite number of components $\Psi_M$.\\

\noindent  Though our model is quite naive, it is interesting to compute explicitly $\frac{\hbar}{ \delta M_\Psi c^2}$ and $\frac{\hbar}{2 E_\Psi}$ in some concrete case. 
 We assume $||A||=1$ supposing that $A$
 is a projection operator: A {\em yes-no} elementary observable. In the following, $a<<b$ can be interpreted in the standard way: The order of magnitude of $b$ is greater than that of $a$.\\  
With $\delta M_\Psi$ comparable with the electron mass and $E_\Psi$ of the order of ground state energy of hydrogen atom we have
$$N^2_\Psi 10^{-20}s << T << N^{-2}_\Psi 10^{-17}s\:.$$ 
It implies that $N_\Psi$ has to be quite small, of the order of the unit in the practise.
 Conversely, increasing the mass, the situation dramatically changes.
For macroscopic values, i.e., $\delta M_\Psi$ of the order of $1Kg$ and $E_\Psi$ of the order of $1J$ we obtain:
$$N^2_\Psi 10^{-50}s << T << N^{-2}_\Psi 10^{-34}s\:.$$ 
Therefore, $N_\Psi$ can be chosen quite large. For instance, values 
$N_\Psi \sim 10^3$ are allowed. \\
The upper bound for $T$ becomes  smaller and smaller as soon as $E_\Psi$ increases.
In particular it happens approaching the  macroscopic classical realm.  However, in that case another effect has to be taken into account, due to the fact that also $\delta M_\Psi$
is supposed to increase and,  in the macroscopic world, the energetic content $Mc^2$ of a mass 
is considerably larger than the typical scales of mechanical energy allowable to the mass.  
  That effect suppresses the part of $R_T$ that is not in agreement with the 
superselection rule.  As a matter of fact, a weaker version of Bargmann's rule turns out to be valid anyway  if referring to temporally averaged quantities on longer periods of time,  regardless
any upper bound for $T$. Indeed the following statement holds essentially relying upon Riemann-Lebesgue's lemma.

\begin{proposition}{\em ({\bf Effective superselction rule II})}
 With the same hypotheses, definitions and notations  as in Prop.\ref{PE1},
one also has for every $A\in \gB(\sH)$:
\beq
 \frac{1}{T} \int_0^T  \langle \Psi| U^*_t A U_t \Psi \rangle  dt \sim  \frac{1}{T} \int_0^T tr\left(\rho_\Psi e^{i\frac{t}{\hbar}\widehat{H}}A_B e^{- i\frac{t}{\hbar}
 \widehat{H}}\right) dt 
\eeq
for
$$2||A||\: N^2_\Psi E_\Psi  << \delta M_\Psi c^2   \quad \mbox{and} \quad \frac{2\hbar  N_\Psi^2 ||A||}{ \delta M_\Psi c^2} << T \:.$$
\end{proposition}

\begin{proof} Again $V_t:= e^{- i\frac{t}{\hbar}\widehat{H}}$.
Distinguishing between the  
the contribution of case $M\neq M'$ and that $M=M'$ in the form of $R_T$:
$$R_T =  \frac{1}{T} \int_0^T  
\left[ \sum_{M} \langle V_t \Psi_M |A V_t \Psi_{M}\rangle -  \sum_{M} \langle  \Psi_M |A \Psi_{M}\rangle \right] dt + R'_T $$ 
$$ =  \frac{1}{T} \int_0^T tr(\rho_\Psi V^*_t A_B V_t) dt  - tr(\rho_\Psi A_B ) + R'_T\:, $$
with
$$ R'_T=\sum_{M \neq M'}\frac{1}{T} \int_0^T  
 e^{i\frac{(M'-M)c^2t}{\hbar}} \left[\langle V_{t}\Psi_M |  A V_{t} \Psi_{M'}\rangle- \langle \Psi_M |  A  \Psi_{M'}\rangle\right] dt\:.
$$
Integrating by parts in the right-hand side,
using Stone theorem and going the same way leading to $|R_T| \leq ||A|| 2T N^2_\psi E_\Psi/\hbar$, one obtains:
$$
|R'_T| \leq \frac{2||A|| N_\Psi^2E_\Psi}{\delta M_\Psi c^2}\:.
$$
Finally, taking (\ref{W}), (\ref{expansion}), and (\ref{one})  into account, the obtained bound  entails the thesis.
\end{proof}

 \section{Discussion}
In this paper we have focussed on  Giulini's model studied in \cite{Giulini,Giulini2} for the mass operator in quantum non-relativistic theories invariant under 
the Galileian group. The reason has been twofold. On the one hand we wished to improve 
 that model from a pure mathematical-physics viewpoint, getting rid of some difficulties (continouous, non positive spectrum of the mass operator). On the other hand, we have drawn on the idea proposed in \cite{Giulini,Giulini2} 
 that some superselection rules, formally obtained imposing some invariance requirement,
  may actually arise from some real dynamical physical process; we have then intended to explicitly discuss
  a, perhaps quite naive, effective de-coherence process responsible for the mass superselection rule.  Both goals have been reached. 
Regarding the first goal, changing the classical extended phase space with respect to that defined in \cite{Giulini}, we have indeed constructed 
a classical model of the mass as a dynamical variable
whose quantization gives rise to a positive point-spectrum for the corresponding self-adjoint operator. This  has been done performing an angle angular-momentum quantization rather than
a position momentum quantization. This was possible in view of the change of the topology of the phase space passing from $\bR$ to $\bS^1$ concerning the domain of the coordinates conjugated with the (signed) masses. 
In our model a further degree of freedom arises that may be interpreted as a conserved charge of the studied particles.
Nevertheless we have remarked that this further degree of freedom should disappear if adopting a quantization procedure based on POVM rather that PVM, like 
the well-known one  exploited in phase number-of particle quantization.\\
  Concerning the second goal we have shown how, assuming 
  that a mass operator exists -- not necessarily with the structure discussed in \cite{Giulini} or in the first part of this work, 
  but with (positive) pure {\em point spectrum} -- the de-coherence process arises just to the 
  presence of the mass operator as a term present in the non-relativistic Hamiltonian. That term is usually and incorrectly 
  neglected  when performing the most simple
  non-relativistic approximation.   
   This effective de-coherence process, when referred to the model of mass operator introduced in the first part of the paper, can alternatively be
  interpreted as an averaging procedure on the unobservable part of the extended  symmetry group $U(1)\times_\omega \cG$. However, for our final result, what is actually the symmetry group of the system -- either $\cG$ or $U(1)\times_\omega \cG$ -- does not seem very relevant from the physical viewpoint.  The Poincar\'e group 
  is, obviously,  the true invariance group  and no constraints on states and observables exist. In our, perhaps naive view, the superselection rule arises as a de-coherence process within the non-relativistic approximation when one takes into account the fact that the mass is an operator, assumed to have a discrete spectrum, and that physical instruments perform a temporal average.
 That process selects allowed states and permitted observables in the following way. Physical measurements cannot distinguish 
 the elements inside certain equivalence classes of states and observables as discussed in Sec. \ref{secclasses} and 
 each of these classes is in fact labelled by, respectively, states and observables allowed by Bargmann superselection rule. Dealing with equivalence classes, namely permitted states and observables, the Galileian group 
 becomes a symmetry group.\\
Alternative approaches are possible, possibly directly leading  to $U(1)\times_\omega \cG$ as a symmetry group 
starting from a group including Poincar\'e one as a subgroup. A candidate may be the conformal group of Minkowski spacetime $SO(4,2)$, following 
the way outlined in Ch.13 of \cite{BR}. \\
We wish eventually stress that  a physical key hypothesis of all this work is the point-wise structure of the spectrum of the mass operator. That hypothesis is supposed to hold
 {\em a posteriori} since it guarantees the validity of our results. However, to author's knowledge, there is no general physical principle leading to that hypothesis 
 {\em a priori}. This is the major gap in our construction and in every similar models.

 \section*{Acknowledgements} The authors are grateful to Alessio Recati  and Claudio Dappiaggi 
 for valuable comments and suggestions.\\
 
\noindent  {\em This work relies upon E.Annigoni's Master Thesis in Mathematics (Trento University a.y. 2010-2011, supervisor V.Moretti), 
 extending some results presented therein.}
\appendix
\section{Proof of Bargmann no-go statements}
Referring to the discussion at the end of Sec.\ref{sec1}, we show that the no-go statement (1) holds and it implies the no-go statment (2). 
Suppose that the unitaries $V_g := \rho_g U^{(M)}_g$ have trivial multipliers. Consider the elements $f=  (0,{\bf c}, {\bf 0},I)$ and $g=  (0,{\bf 0}, {\bf v},I)$, by direct inspection find:
$(U^{(M)}_{f})^{-1}(U^{(M)}_{g})^{-1}U^{(M)}_{f} U^{(M)}_{g}= e^{-2iM {\bf c}\cdot {\bf v}}I$
that implies $V_{f^{-1}}V_{g^{-1}}V_{f} V_{g}= e^{-2iM {\bf c}\cdot {\bf v}}I$.
  However, since $f^{-1}g^{-1}f g= e$ and the multipliers of $V_h$ are trivial, we also have
$V_{f^{-1}}V_{g^{-1}}V_{f} V_{g}= I$, that is in contradiction  with what we have already found unless $M=0$ that is not allowed.
So (1) holds true. Now, for $M>M'$, suppose that there are phases such that
$\rho'_g U^{(M)}_g = \rho_g U^{(M')}_g$ for every $g\in \cG$. Redefining $\rho_g^{-1}\rho'_g$ as 
$\rho_g$, we have:
$U^{(M')}_g = \rho_g U^{(M)}_g$ for every $g\in \cG$. In terms of multipliers 
 it implies
\beq  \omega^{(M-M')}(g',g)= \frac{\rho_{g'g}}{\rho_g\rho_{g'}}\quad \mbox{for all $g,g'\in \cG$}\:.
\eeq
This identity immediately implies that  the mulipliers of $V_g := \rho_g U^{(M-M')}_g$ are 
trivial, but it is impossible as established above. So (2) holds as well.
\section{Proof of some propositions}
\begin{proof} ({\bf Lemma \ref{propstates}}.)\\
The only thing to be proved is that $\rho_\Psi$ is trace class. The convergence of the orthogonal series decomposing  $\Psi$ in its components in every $\sH_M$
 easily implies that the series of $\rho_\Phi$ converges in the uniform topology so that $\rho_\Psi$ is compact.
$\rho_\Psi$ is trace class (with $tr\rho_\Psi = ||\Psi||^2 =1$) since it is positive and thus $|\rho_\Psi|= \rho_\Psi$, in view of the very definition of the latter, admits trace on a Hilbert basis
obtained by completing $\{\Phi_M\}_{M\in S}$.
\end{proof}

\begin{proof} ({\bf Theorem \ref{propobs}}.)\\
Notice that (\ref{tr}) immediately follows from (\ref{BA}). Linearity of $\gB(\sH) \ni A \mapsto A_B \in \gA_B$ is obvious. 
Let us prove the latter together with (ii). Since $\{P_M AP_M\Psi\}_{M\in \sigma(\widehat{M})}$ 
is an orthogonal set:
$$||\sum_{M\in \sigma(\widehat{M})} P_M AP_M\Psi||^2 = \sum_{M\in \sigma(\widehat{M})} ||P_M AP_M\Psi||^2 \leq 
||A||^2 \sum_{M\in \sigma(\widehat{M})} ||P_M\Psi||^2 \leq ||A||^2||\Psi||^2<+\infty$$
so $A_B$ is defined on the whole $\sH$ and $||A_B||\leq ||A||$. The item (i) is now obviously true. (iii) arises using the weak convergence of the 
series, the uniqueness
of the adjoint operator and the self-adjointness of the $P_M$. (iv) holds because  $A_B$ is obviously positive if $A$ is, $||(A_B)^{1/2}||^2 = ||A_B||\leq ||A|| $
and $||A||=1$ since $A$ is an orthogonal projector and, finally, $||(A_B)^{1/2}|| \leq 1$ implies $\langle \Psi |A_B \Psi \rangle \leq \langle 
\Psi |\Psi \rangle$, that is $A_B\leq I$. Let us prove (v). If $A_B=A'_B$ then $tr(\rho A)= tr(\rho A')$ for all $\rho \in \gS_B$. The converse implication,
using linearity and (iii) and decomposing any element of $\gB(\sH)$ into self-adjoint and anti self-adjoint part,  immediately arises if, 
for $A=A^*\in \gB(\sH)$, $tr(\rho A)= 0$ for all $\rho \in \gS_B$ entails 
$A_B=0$. Let us prove that it holds true.  Assume that $tr(\rho A)= 0$ for all $\rho \in \gS_B$ for a given $A=A^*$. Fix $P_M$. 
By hypotheses, as every element $\Psi_M \in \sH_M$ can be completed as a basis of that space and used to compute the trace, 
taking $\rho = \Psi_M\langle \Psi_M|\cdot \rangle$ one has that it must be $\langle \Psi_M | P_MAP_M \Psi_M\rangle =0$.
Since $\Psi_M$ is arbitrary and $P_MAP_M$ is self-adjoint, it easily implies (using the polarization identity) 
that $\langle \Psi_M | P_MAP_M \Psi'_M\rangle =0$ for every pair $\Psi_M, \Psi'_M \in \sH_M$, so that $P_MAP_M=0$
and thus $A_B=0$.
\end{proof}

 \begin{proof} ({\bf Theorem \ref{teotot}}.)\\ 
 We start by proving that $\rho_B \in \gS_B$. By construction $\rho_B \geq 0$ when $\rho\geq 0$. Referring to a Hilbert basis  of 
 eigenvectors of $\widehat{M}$,
$\{\psi_{M,k}\}$ with $\psi_{M,k} \in \sH_M$,
 we have that:
 $$\sum_{M, k} \langle \psi_{M,k} |\rho_B\psi_{M,k}\rangle =  \sum_{M, k} \langle \psi_{M,k} |\rho\psi_{M,k}\rangle = tr(\rho) =1 < +\infty\:.$$
 Therefore $\rho_B \in \gS_B$. With the same argument (\ref{trtr}) easily arises. 
The statement  (i) is equivalent to $(\rho-\rho')_B=0$ iff $tr((\rho-\rho')A)=0$ for every $A\in \gA_B$. 
The proof of this statement is the same as of (v) of theorem  \ref{propobs}, taking $A= \Psi_M \langle \Psi_M| \cdot \rangle$
for every $P_M$ and every $\Psi_M \in \sH_M$.
 The last statement is obvious.
 \end{proof}

\end{document}